\documentclass[12pt]{article}

\usepackage[colorlinks=true, linkcolor=red, urlcolor=blue, citecolor=blue, anchorcolor=blue]{hyperref}     
\usepackage{graphicx}
\usepackage[shortlabels]{enumitem}

\usepackage[top=2cm, bottom=2cm, left=2.8cm,
right=2.8cm]{geometry} 

\usepackage{setspace}

\usepackage{natbib}
\usepackage{url}

\usepackage[most]{tcolorbox}
\newcounter{mytcbcounter}
\newtcolorbox[use counter=mytcbcounter]{mybox}[2][]{
    float,
    title={\textbf{\textcolor{black}{Algorithm~\themytcbcounter: #2}}},
    #1, 
    colframe=white!70!gray
    }

\usepackage{notation}

\crefname{theorem}{Theorem}{Theorems}
\crefname{figure}{Figure}{Figures}
\crefname{equation}{}{}
\crefname{definition}{Definition}{Definitions}
\crefname{assumption}{Assumption}{Assumptions}
\crefname{corollary}{Corollary}{Corollaries}
\crefname{proposition}{Proposition}{Propositions}
\crefname{remark}{Remark}{Remarks}
\crefname{principle}{Principle}{Principles}
\crefname{lemma}{Lemma}{Lemmata}
\crefname{claim}{Claim}{Claims}
\crefname{table}{Table}{Tables}
\crefname{section}{Section}{Sections}
\crefname{subsection}{Section}{Sections}
\crefname{subsubsection}{Section}{sections}
\crefname{assumption}{Assumption}{Assumptions}
\crefname{appendix}{Appendix}{Appendices}
\numberwithin{equation}{section}
\crefname{condition}{Condition}{Conditions}

\title{Post-selection inference for penalized M-estimators \\  via score thinning}
\author{
Ronan Perry$^{1}$\thanks{Corresponding author: rflperry@uw.edu}, Snigdha Panigrahi$^2$, Daniela Witten$^{1,3}$
\\\\
\small
$^1$ University of Washington, Department of Statistics\\\small
$^2$ University of Michigan, Department of Statistics\\\small
$^3$ University of Washington, Department of Biostatistics
}
\date{\today}

\begin{document}

\maketitle

\doublespacing

\abstract{
We consider inference for M-estimators after model selection using a sparsity-inducing penalty. While existing methods for this task require bespoke inference procedures,  we propose a simpler approach, which relies on two insights:
(i) adding and subtracting  carefully-constructed  noise to a Gaussian random variable with unknown mean and known variance leads to  two \emph{independent} Gaussian random variables; and (ii) both the selection event resulting from penalized M-estimation, and the event that a standard (non-selective) confidence interval for an M-estimator covers its target, can be characterized in terms of an approximately normal ``score variable". We combine these insights to show that --- when the noise is chosen carefully --- there is asymptotic independence between the  model selected using a noisy penalized M-estimator, and the event that  a standard (non-selective) confidence interval on noisy data covers the selected parameter. Therefore, selecting a model via penalized M-estimation (e.g. \texttt{glmnet} in \texttt{R}) on noisy data, and then conducting \emph{standard} inference on the selected model (e.g. \texttt{glm} in \texttt{R}) using noisy data, yields valid inference: \emph{no bespoke methods are required}. Our results require independence of the observations, but only weak distributional requirements.
We apply the proposed approach to conduct inference on the association between sex and smoking in a social network.
}
\noindent%

\maketitle

\section{Introduction}

In classical statistics, we  collect  data $Y$ in order to conduct inference on  a pre-determined target parameter $\theta_0$. However, in practice, a researcher may use the data to select a data-driven parameter $\theta_0(Y)$, and then conduct inference using the same data $Y$; this is  sometimes referred to as \textit{double dipping}~\citep{kriegeskorte_circular_2009}. In general, if we treat $\theta_0(Y)$ \textit{as if it were fixed} and use classical inference techniques, then standard inferential guarantees (such as Type 1 error control of p-values, and coverage of confidence intervals) will not apply \citep{hotelling1940selection, cox_note_1975, benjamini_selective_2009, berk_valid_2013, fithian_optimal_2017, kuchibhotla2022post, neufeld_inference_2026}. 

In recent years, particular interest has focused on inference on models selected via sparsity-penalized M-estimation, such as the  lasso~\citep{tibshirani1996regression}. The goal is to develop p-values that  control the  Type 1 error and confidence intervals that attain the nominal  coverage \emph{conditional on the selected model}  \citep{fithian_optimal_2017}. We briefly review this line of work in Section~\ref{sec:relatedwork}, and then summarize the contributions of our paper in Section~\ref{sec:contributions}.

\subsection{Existing approaches for inference after model selection}\label{sec:relatedwork}

First, we review methods that are specifically tailored for inference on a model selected via penalized M-estimation.  
The  conditional selective inference proposal of \citet{lee_exact_2016} provides valid inference on linear contrasts of the mean of a Gaussian response vector when these contrasts were selected by the lasso; this is accomplished by characterizing the lasso selection event as a set of linear inequalities, and then deriving the distribution of the estimator conditional on the lasso selection event. 
\citet{tibshirani_uniform_2018} and \citet{tian_asymptotics_2017} extend these results to provide asymptotic guarantees without distributional assumptions on the response vector, and \citet{tian_asymptotics_2017}  allow the number of features $p$ to grow with the sample size $n$. These methods provide valid inference conditional on the event that a specific model was selected via the lasso using all of the data. However, this framework requires analytical characterization of the selection event  for each model selection procedure; furthermore, it lacks power near the boundary of the selection event~\citep{kivaranovic_length_2021}.

To improve power, randomized conditional selective inference methods introduce noise into the selection process. In their pioneering work, \citet{tian_selective_2018} prove a distribution-free selective central limit theorem which they use to derive a selective density by conditioning on having selected a model using an $L_1$-penalty. However, this selective density must be approximated via sampling. In a linear Gaussian model, \citet{panigrahi_exact_2024} derive a numerically-tractable closed-form distribution for inference after the randomized lasso, while \citet{bakshi_selective_2024} extend this to $L_1$-penalized M-estimators. In the context of group-lasso-penalized M-estimators, \citet{huang_selective_2024} derive a selective density which can be numerically approximated. Once again, these approaches require the derivation of a bespoke method for inference for each model selection procedure.
 
We now turn to approaches that split the data into two independent sets, so that the first can be used to select a model and the second to conduct inference. A major advantage of such approaches over the aforementioned (randomized) conditional selective inference approaches is that no bespoke methods for inference are required, i.e., valid inference can be obtained using standard inference methods and software, e.g., \verb=glm= and \verb=glmnet= in \verb=R= \citep{friedman_regularization_2010}.
Provided that multiple independent and identically-distributed random variables are available, sample splitting~\citep{cox_note_1975} is a straightforward option. However, if the observations are not independent then sample splitting is not applicable. Even if they are independent, sample splitting can be unattractive if they are not identically distributed, or if the sample size is very small~\citep{rasines_splitting_2023, dharamshi_generalized_2025}. 
Recent proposals have moved beyond sample splitting: \citet{robins_adaptive_2006} show that a single Gaussian random variable $Y$ with unknown mean and known variance can be split into two independent Gaussians $Y^{(1)}$ and $Y^{(2)}$; this fact was subsequently applied to construct confidence intervals~\citep{tian_selective_2018, rasines_splitting_2023} and to estimate prediction error~\citep{tian_prediction_2020, oliveira_unbiased_2024} after model selection. \citet{leiner_data_2025} further considered a similar decomposition in the Poisson case, while \citet{neufeld_data_2024} and \citet{dharamshi_generalized_2025} generalize this idea to a much broader set of distributions and refer to this generalization as ``data thinning''. 
  
The distributional assumptions required for data thinning are restrictive, and recent work seeks to loosen them in exchange for asymptotic (rather than finite-sample) guarantees for i.i.d. data. \citet{rasines_splitting_2023} show that it is possible to thin the data $Y$ as though it were Gaussian and still obtain asymptotically valid inference on linear contrasts of $\EE[Y^{(2)}]$ when the selection event is that $Y^{(1)}$ lives in a convex set. \citet{lei_discussion_2025} provides asymptotically valid inference on linear contrasts of $\EE[Y^{(2)}]$ when the selection event is any function of the contrasts $X^\top Y^{(1)}$ under certain conditions, e.g., (i) when the selection event is that $X^\top Y^{(1)}$ lives in a convex set, or (ii) a total-variation bound requiring continuous errors $Y - \EE[Y]$; notably, these conditions do not hold in the case of logistic regression. Separately, \citet{panigrahi_exact_2024} and \citet{fry_assumption-lean_2025} demonstrate how to asymptotically thin an estimator, e.g., the OLS estimator, using any available bootstrap procedure with a CLT. However, this does not enable inference on a selected model; instead, it enables inference on a selected parameter in a \emph{fixed} model.

Finally, it bears mention that a line of work by \citet{meinshausen_p-values_2009, zhang_confidence_2014, javanmard_confidence_2014, sengupta_ell-test_2025} uses the lasso as a tool to conduct inference on an underlying sparse, high-dimensional linear model. Thus, their target parameter is \emph{not} a function of the data.

\subsection{Contributions of this paper}\label{sec:contributions}

In this paper, we consider inference on a parameter in a model selected via  penalized M-estimation, e.g., the lasso. We do not make distributional assumptions, nor do we assume identically distributed observations. 
Our developments rely upon a key insight: namely, that both selection and inference hinge upon the same approximately normal \emph{score variable}, $\hat Z_n$.

In greater detail, our insights  are as follows:
\begin{description}[leftmargin=0cm]
   \item[Insight 1: ] Since $\hat Z_n$ is approximately normal, it can be decomposed into two pieces, $\hat Z_n^{(1)}$ and $\hat Z_n^{(2)}$, which are approximately normal and independent, up to appropriately vanishing remainders.
    \item[Insight 2: ] The event that a particular  model is selected via ``noisy" penalized M-estimation is equivalent to an event involving $\hat{Z}_n^{(1)}$. Furthermore, the event that a standard (non-selective) confidence interval  centered around a ``noisy" M-estimator covers its target is equivalent to an event involving $\hat{Z}_n^{(2)}$. 
    \item[Insight 3: ]  Combining Insights 1 and 2, if we select a model using a ``noisy" penalized M-estimator, then  a standard confidence interval centered around a ``noisy" M-estimator attains the nominal coverage, conditional on the selection event. 

\end{description}

Thus, \emph{we can conduct valid post-selection inference for models selected via (noisy) penalized M-estimation using standard inference pipelines.} This is in marked contrast to the methods discussed in Section~\ref{sec:relatedwork}, which require bespoke inference techniques \citep{lee_exact_2016, fithian_optimal_2017, tian_asymptotics_2017, tibshirani_uniform_2018, tian_selective_2018, panigrahi_exact_2024, huang_selective_2024, leiner_data_2025}, stringent distributional assumptions \citep{neufeld_data_2024, dharamshi_generalized_2025}, or estimands that are limited to linear contrasts of i.i.d. outcomes \citep{rasines_splitting_2023, lei_discussion_2025}.

The rest of this paper is organized as follows.
In \cref{sec:selection_and_inference}, we formalize Insight 1.
In \cref{sec:valid_inf}, we formalize Insights 2 and 3. 
In \cref{sec:glms}, we show that our theory can be applied to conduct inference on parameters in generalized linear models selected via $L_1$ penalization. In \cref{sec:applications}, we examine the validity of our method in simulations and apply it to conduct inference after principal components network regression on the \emph{Teenage Friends and Lifestyle Study} dataset. \cref{sec:discussion} concludes with a discussion.

\section{Valid post-selection inference via a score variable}\label{sec:selection_and_inference}

We review the M-estimation setting in Section~\ref{subsec:m_est}, present the score variable $\hat Z_n$ in Section~\ref{subsec:core_stat}, and discuss the construction and properties  of  $\hat Z_n^{(1)}$ and $\hat Z_n^{(2)}$ in Section~\ref{subsec:thin_core}. 

\subsection{M-estimation}\label{subsec:m_est}

Suppose that $Y_1, \dots, Y_n \overset{\mathrm{ind.}}{\sim} F_n$ form a triangular array of independent random variables $Y_i \in \Ycal \subseteq \RR^d$ for some $d > 0$, keeping the dependence of the distribution of $Y_i$ on $n$ implicit. Let $\ell_\theta : \Ycal \mapsto \RR$ be a differentiable loss function parameterized by $\theta \in \RR^p$. We will implicitly allow this loss function to depend on the index $i$, e.g., in the case of fixed-X regression where fixed covariates $X_i$ may be available.

Let $E \subseteq [p]$ denote the indices of any non-empty subset of the parameters. Without loss of generality, we assume that $E = \cbr{1, 2, \dots, |E|}$. Our goal is to conduct inference on the population estimand corresponding to the  model with features indexed by $E$, 
\begin{equation}\label{eq:MstarE}
    \MstarE = \argmin_{\theta \in \RR^{\abr{E}}} \cbr{\frac{1}{n} \sum_{i=1}^n \EE \left[ \ell_{\rbr{\theta, 0_{p-\abr{E}}}} (Y_i) \right] }= \argmin_{\theta \in \RR^{\abr{E}}} \cbr{P_n \EE \ell_{\rbr{\theta, 0_{p-\abr{E}}}}},
\end{equation}
where  the last equality makes use of shorthand that will appear throughout this paper: namely,  $P_n f:=\tfrac1n \sum_{i=1}^n f(Y_i)$, for any function $f$ defined over $\Ycal$.

For fixed model $E$, level $\alpha \in (0, 1)$, and contrast $\xi \in \RR^{|E|}$, consider the function $\CI_n^{\xi,\alpha}: \rbr{\mathbb{R}^{|E|} \times \mathbb{S}_{|E|}^{++}} \mapsto \mathbb{R}^2$, where $\mathbb{S}_{|E|}^{++}$ denotes the space of $|E| \times |E|$ positive definite matrices. For any $\theta^E \in \mathbb{R}^{|E|}$ and $S^E \in  \mathbb{S}_{|E|}^{++}$, we define
\begin{equation}\label{eq:classical_ci}
    \CI^{\xi, \alpha}_n\rbr{\theta^{E}, S^E} := \sbr{\xi^\top \theta^E \pm n^{-1/2}z_{1-\alpha/2}\rbr{\xi^\top S^E \xi}^{1/2}},
\end{equation}
where $z_{1-\alpha/2}$ is  the $(1-\alpha/2)$-quantile of a standard normal distribution. Then, 
for the M-estimator $\tilde \theta^E_n \in \mathbb{R}^{|E|}$ based on the model $E$ and defined as
\begin{equation}\label{eq:MestEnaive}
    \MestEnaive :=
    \argmin_{\theta \in \RR^{\abr{E}}} \cbr{P_n \ell_{\rbr{\theta, 0_{p-\abr{E}}}}},
\end{equation}
and for the sandwich variance estimator $\tilde{S}_n^E := \sbr{P_n \ddot\ell_{\MestEnaive}}^{-1} P_n \sbr{\dot\ell_{\MestEnaive}\dot\ell_{\MestEnaive}^\top} \sbr{P_n \ddot\ell_{\MestEnaive}}^{-1}\in \mathbb{S}^{++}_{|E|}$, it holds that $\CI^{\xi, \alpha}_n\rbr{\MestEnaive, \Shatnaive}$ is a valid confidence interval under typical conditions \citep[Section 5.2.1]{vaart_asymptotic_1998}, in the sense that it attains the nominal (asymptotic) coverage, i.e., 
\begin{equation}\label{eq:naive_coverage}
    \lim_{n \to \infty}\Pr\rbr{\xi^\top \MstarE \in \CI^{\xi, \alpha}_n\rbr{\MestEnaive, \Shatnaive}} \geq 1-\alpha.
\end{equation}

Now, suppose instead that $E$ is random: for instance, it may be the support of  a penalized M-estimator, i.e., $E:= \supp\rbr{\tilde \theta_n^{(\lambda)}}$,  where 
\begin{equation}\label{eq:MestPnaive}
    \MestPnaive
   \in \argmin_{\theta \in \RR^p} \cbr{P_n \ell_\theta + n^{-1/2} \rho_\lambda(\theta)},
\end{equation}
and $\rho_\lambda : \RR^p \mapsto \RR$ is a non-differentiable penalty parameterized by a fixed tuning parameter $\lambda$. For example, using the $L_1$-penalty $\rho_\lambda(\theta) := \lambda \nbr{\theta}_1$ and squared-error loss $\ell_\theta(X_i, Y_i) := \frac{1}{2} (Y_i - X_i^\top \theta)^2$, \cref{eq:MestPnaive} is the popular lasso estimator \citep{tibshirani1996regression}. When $E$ is random, the interval $\CI^{\xi, \alpha}_n\rbr{\MestEnaive, \Shatnaive}$ is no longer guaranteed to achieve the nominal coverage, i.e., \eqref{eq:naive_coverage} does not hold. 
Our goal in this paper is to construct estimators  $\hat \theta^E_n$ and $\Shat$ such that the interval  $\CI^{\xi, \alpha}_n\rbr{\MestE, \Shat}$ attains the nominal coverage, \emph{conditional} on the fact that  $E$ is the support of some penalized M-estimator $\MestP$ that remains to be specified. That is, we wish for the conditional coverage guarantee
\begin{equation}\label{eq:conditional_coverage}
    \lim_{n \to \infty}\Pr\rbr{\xi^\top \MstarE \in \CI^{\xi, \alpha}_n\rbr{\MestE, \Shat}~\middle|~\supp\rbr{\MestP} = E} \geq 1-\alpha.
\end{equation}

\subsection{The score variable}\label{subsec:core_stat}

It should not be surprising that since \cref{eq:MestEnaive} and \cref{eq:MestPnaive} are solutions to a minimization problem, the gradient of the loss function $\ell_\theta$ with respect to $\theta$ plays an important role. In fact, this gradient defines the score variable.
Recalling the definition of $\MstarE$  in \cref{eq:MstarE}, we define
\begin{equation}\label{eq:Mstar}
    \Mstar := (\MstarE, 0_{\abr{E^c}}),
\end{equation}
the expanded $p$-dimensional vector. For any $\theta$, define the gradient
\begin{equation}\label{eq:first_deriv_def}
    \dot\ell_{\theta}(Y_i) := \sbr{\frac{\partial}{\partial \theta_1'} \ell_{\theta'}(Y_i), \dots, \frac{\partial}{\partial \theta_p'} \ell_{\theta'}(Y_i)}^\top \Bigg\rvert_{\theta' = \theta} \in \RR^{p}.
\end{equation}
Under standard conditions for M-estimation (in which $E$ is fixed), the $E$-indices of the score variable
\begin{equation}\label{eq:core_stat}
    \hat Z_n := \sqrt{n} P_n \dot\ell_{\Mstar}
\end{equation}
are asymptotically normal; this property forms the basis for asymptotically valid inference. However, here our interest lies in asymptotically valid inference when the set $E$ arises from model selection. In this setting, we will show that $\hat Z_n$ is notable in that it contains all of the information needed for both selection \emph{and} inference, up to a vanishing remainder.

\subsection{Approximate normality of the thinned score variable} \label{subsec:thin_core}

We will now borrow ideas from 
the literature on data thinning~\citep{rasines_splitting_2023, neufeld_data_2024, dharamshi_generalized_2025} to decompose the score variable $\hat Z_n$ in \cref{eq:core_stat} into two components, which we will then show are approximately normally distributed, in the sense of a Berry-Esseen-type bound \citep{raic_multivariate_2019}, and approximately conditionally independent, in a sense that we will specify. This formalizes Insight 1 in Section~\ref{sec:contributions}. 

Using the notation defined in \cref{eq:first_deriv_def}, we require the following conditions. Let $\mu_{n, i}^{E, *} := \EE \sbr{\dot\ell_{\Mstar}(Y_i)}$ and $\mustar := \frac1n \sum_{i=1}^n \mu_{n, i}^{E, *}$.

\begin{condition}\label{cond:berry_esseen}
    \,
    \begin{enumerate}[(i)]
        \item $\Sigma_n := \tfrac1n \sum_{i=1}^n\Var\rbr{\dot\ell_{\Mstar}}$ is invertible.
        \item $P_n \EE \nbr{\dot\ell_{\Mstar} - \mu_{n, i}^{E, *}}_2^3$ is uniformly bounded above.
    \end{enumerate}
\end{condition}

Under \cref{cond:berry_esseen}, it can be shown that $\hat Z_n$ is approximately $\Norm(\sqrt{n} \mustar, \Sigma_n)$ \citep{raic_multivariate_2019}. We now apply Gaussian data thinning to $\hat{Z}_n$ as though it were exactly normally distributed. That is, for some $\hat \Sigma_n \in \RR^{p \times p}$, we define $\hat W_n \sim \Norm(0, \hat \Sigma_n)$ and
\begin{equation}\label{eq:define_thins}
    \hat Z_n^{(1)} := \hat Z_n + \gamma \hat W_n \quad\text{and}\quad \hat Z_n^{(2)} := \hat Z_n - \frac1\gamma \hat W_n.
\end{equation}
The tuning parameter $\gamma$ determines the amount of noise added to $\hat{Z}_n$ to construct $\hat{Z}_n^{(1)}$ versus $\hat{Z}_n^{(2)}$. When $\gamma = 1$, the same amount of noise is added to each, and the procedure evokes splitting $n$ samples into two equal pieces.

The following Berry-Esseen-type bound, using results from \citet{raic_multivariate_2019}, establishes that $\rbr{\hat Z_n^{(1)}, \hat Z_n^{(2)}}$ are approximately distributed as $\rbr{Z_n^{(1)}, Z_n^{(2)}}$, where
\begin{equation}\label{eq:define_thins_normal}
    Z_n^{(1)} \sim \Norm\rbr{\sqrt{n}\mustar, \rbr{1 + \gamma^2}\Sigma_n} \quad\text{and}\quad Z_n^{(2)} \sim \Norm\rbr{\sqrt{n}\mustar, \rbr{1 + \tfrac{1}{\gamma^2}}\Sigma_n}
\end{equation}
are independent random variables.  Let $\eigmin{\cdot}$ denote the minimum eigenvalue of its argument, a positive semi-definite matrix.

\begin{lemma}\label{lemma:thinning_core_stat}
Suppose that $C \subseteq \RR^{2p}$ is the union of $K$ measurable and almost surely disjoint convex sets. 
Under Condition \ref{cond:berry_esseen}, for any $\Delta \geq 0$, 
\begin{align}
    &\abr{
    \Pr\rbr{\rbr{\hat Z_n^{(1)}, \hat Z_n^{(2)}} \in C} - \Pr\rbr{\rbr{Z^{(1)}_n, Z^{(2)}_n} \in C}
    }\nonumber\\
    &\leq
    \Pr\rbr{\nbr{\hat \Sigma_n - \Sigma_n}_F > \Delta \eigmin{\Sigma_n}}
    + O\rbr{K p^{1/4} n^{-1/2} \eigmin{\Sigma_n}^{-3/2} + \Delta}
    .\label{eq:thinning_core_stat}
\end{align}
\end{lemma}

When $\hat \Sigma_n = \Sigma_n$, i.e., the covariance $\hat \Sigma_n$ used to construct $\hat Z_n^{(1)}$ and $\hat Z_n^{(2)}$ in \cref{eq:define_thins} exactly equals $\Sigma_n$ defined in \cref{cond:berry_esseen}, then we can set $\Delta = 0$ to obtain the classical Berry-Esseen-type rate. We now establish a conditional Berry-Esseen-type result.

\begin{lemma}\label{lemma:conditional_berry_esseen}
    Let $A_n$ and $B_n$ be sets such that $A_n \times B_n$ is the union of $K$ measurable and almost surely disjoint convex sets, and let $r_n^A$ and $r_n^B$ be random vectors in $\mathbb{R}^p$. For any $\Delta \geq 0$, and $\varepsilon > 0$ smaller than the smallest inradius of those $K$ sets \cref{eq:inradius},
    define the rates
    \begin{align*}
        E_1 &:= p^{1/4} n^{-1/2}\eigmin{\Sigma_n}^{-3/2} + \Delta + \Pr\rbr{\nbr{\hat \Sigma_n - \Sigma_n}_F > \Delta \eigmin{\Sigma_n}},\\
        E_2 &:= O\rbr{E_1 + p^{1/2}\varepsilon\rbr{\eigmin{\Sigma_n}^{1/2} + \eigmin{\Sigma_n}^{-1/2}}} + \Pr\rbr{\nbr{r_n^A}_2 > \tfrac\varepsilon2} + \Pr\rbr{\nbr{r_n^B}_2 > \tfrac\varepsilon2},\\
        E_3 &:= \Pr\rbr{Z^{(1)}_n \in A_n^{-\varepsilon}} - O\rbr{E_1} - \Pr\rbr{\nbr{r_n^A}_2 > \varepsilon},
    \end{align*}
    where $A_n^{-\varepsilon}$ is the $\varepsilon$-shrinkage of the set $A_n$, as defined in \cref{eq:epsilon_shrinkage}.
    Under Condition \ref{cond:berry_esseen},
    \begin{align}\label{eq:master_stat}
        \abr{
        \Pr \rbr{\hat Z^{(2)}_n + r^B_n \in B_n \middle| \hat Z^{(1)}_n + r^A_n \in A_n}
        - \Pr \rbr{Z^{(2)}_n \in B_n}
        }
        \leq
        \frac{
        K E_2}
        {\max\cbr{E_3, K E_2}}.
    \end{align}
\end{lemma}

In the next section, we will show that 
\cref{lemma:conditional_berry_esseen} enables valid post-selection inference on a model selected via penalized M-estimation. 

\section{Asymptotically valid post-selection inference}\label{sec:valid_inf}

In Sections~\ref{sec:noisypenalty} and \ref{sec:noisyresponse}, we apply \cref{lemma:conditional_berry_esseen} to establish  two approaches for valid conditional selective inference in the asymptotic regime where $p$ is fixed and $n$ grows.

\subsection{Post-selection inference by adding a noisy penalty} \label{sec:noisypenalty}

Suppose that $\gamma > 0$ and $\hat W_n \sim \Norm(0, \hat \Sigma_n)$ for any covariance matrix $\hat \Sigma_n$.
We consider inference on the selected model $\supp\rbr{\MestP}$, where $\MestP$ is the solution to a penalized M-estimation problem, i.e.    
\begin{equation}\label{eq:MestP}
    \MestP = \argmin_{\theta \in \RR^p} \cbr{P_n \ell_\theta + \gamma n^{-1/2} \hat W^\top_n \theta + n^{-1/2} \rho_\lambda\rbr{\theta}}.
\end{equation}
Note that \eqref{eq:MestP} is a variant of \cref{eq:MestPnaive}, with the addition of the noisy penalty term $\gamma n^{-1/2} \hat W^\top_n \theta$; it resembles randomized estimators seen in the conditional selective inference literature \citep{tian_selective_2018, panigrahi_exact_2024, huang_selective_2024}.  
 
We will further consider the randomized M-estimator on the sub-model with features in the index set $E \in \{1,\ldots,p\}$, where $\hat W_{n, E}$ is the $|E|$-dimensional vector obtained by subsetting the elements of $\hat W_n$ to those indexed by $E$:
\begin{equation}\label{eq:MestE}
    \MestE :=
    \argmin_{\theta \in \RR^{\abr{E}}} \cbr{P_n \ell_{\rbr{\theta, 0_{p-\abr{E}}}} - \tfrac1\gamma n^{-1/2} \hat W^\top_{n, E} \theta}.
\end{equation}
We note that \eqref{eq:MestE} is a variant of  \cref{eq:MestEnaive} that incorporates the added noisy penalty term $-\tfrac1\gamma n^{-1/2} \hat W^\top_n \theta $. As in $\Mstar$, we define the expanded $p$-dimensional minimizer
\begin{equation}\label{eq:Mest}
    \Mest := \rbr{\MestE, 0_{\abr{E^c}}}.
\end{equation}

We require these estimators to be unique almost surely.
\begin{condition}\label{cond:unique_solution}
    The solutions to \cref{eq:MestP} and \cref{eq:MestE} are unique almost surely. 
\end{condition}

Furthermore, we require the existence of a valid \emph{fixed model} confidence interval \cref{eq:classical_ci} around our point estimator \cref{eq:MestE}.
\begin{condition}\label{cond:sandwich_est}\,
    There exists an estimator $\Shat \in \mathbb{S}_{|E|}^{++}$ satisfying $$\lim_{n \to \infty} \Pr\rbr{\xi^\top \MstarE \in \CI_n^{\xi, \alpha}\rbr{\MestE, \Shat }} \geq 1-\alpha.$$ 
\end{condition}

Before delving further into technical details, we outline the key intuition of this section. Recall the construction of $\hat{Z}_n^{(1)}$ and $\hat{Z}_n^{(2)}$ from \eqref{eq:define_thins}. First, in \cref{lemma:selection_event}, we will show that for any index set $E \subseteq \{1,\ldots,p\}$, the event  $\left\{ \supp\rbr{\MestP} = E \right\}$ equals  the event that $\hat{Z}_n^{(1)}$, plus a small remainder term, belongs to some union of convex sets. Then, in \cref{lemma:inference_event}, we will show that the event   $ \left\{  \xi^\top \MstarE \in \CI_n^{\xi, \alpha}\rbr{\MestE, \Shat}  \right\}$ equals the event that  $\hat{Z}_n^{(2)}$, plus a small remainder term, belongs to some convex set, where $\MstarE$ was defined in \eqref{eq:MstarE}. Together, Lemmas~\ref{lemma:selection_event} and \ref{lemma:inference_event} will formalize Insight 2 from Section~\ref{sec:contributions}. Combining these results with Lemma~\ref{lemma:conditional_berry_esseen} will enable us to show, in Theorem~\ref{thm:valid_inf_grad}, that the event $ \left\{  \xi^\top \MstarE \in \CI_n^{\xi, \alpha}\rbr{\MestE, \Shat}  \right\}$ holds with probability $(1-\alpha)$ asymptotically, conditional on the selection event $\left\{ \supp\rbr{\MestP} = E \right\}$; this formalizes Insight 3.

We first require that our loss function satisfies certain typical M-estimation properties, namely that it converges uniformly across the parameter space and can be well-approximated by a second-order Taylor expansion. This approximation is well-studied in the M-estimation literature and will be discussed in further detail in \cref{sec:glms} in the context of generalized linear models.
\begin{condition}\label{cond:taylor_expansion}
    Our estimators and loss function admit the following Taylor expansions centered at $\Mstar$ defined in \eqref{eq:Mstar}: 
    \begin{equation*}
        P_n \dot\ell_{\MestP} = P_n \dot\ell_{\Mstar} + H_n \rbr{\MestP - \Mstar} + o_p\rbr{n^{-1/2}}
    \end{equation*}
    and
    \begin{equation*}
        P_n \dot\ell_{\Mest} = P_n \dot\ell_{\Mstar} + H_n \rbr{\Mest - \Mstar} + o_p\rbr{n^{-1/2}},
    \end{equation*}
    for some deterministic matrix $H_n \in \mathbb{S}_{p}^{++}$, and where $\dot\ell_\theta$ was defined in \cref{eq:first_deriv_def}.
\end{condition}

Furthermore, we assume that the following population quantities converge asymptotically to non-degenerate quantities. 

\begin{assumption}\label{ass:existing_limits}
   Recalling the definitions of $\MstarE$ in \cref{eq:MstarE}, $\Sigma_n$  in \cref{cond:berry_esseen}, and $H_n$  in \cref{cond:taylor_expansion}, the following limits exist:
    \begin{enumerate}[(i)]
        \item $\sqrt{n}\MstarE \to \theta^{E, *}$ for some fixed and bounded $\theta^{E, *}$,
        \item $\Sigma_n \to \Sigma$ for some fixed $\Sigma \in \mathbb{S}_{p}^{++}$,
        \item $H_n \to H$ for some fixed $H \in \mathbb{S}_{p}^{++}$.
    \end{enumerate}
\end{assumption}
In our theoretical development, \cref{ass:existing_limits}(i) will guarantee that the probability of selecting the model $E$ is asymptotically bounded above zero. Similar local alternative-type assumptions are made in related work \citep{tian_selective_2018, tibshirani_uniform_2018}. In contrast, \citet{lei_discussion_2025} guarantee only unconditional selective coverage and thus do not require such an assumption, 
while \citet{rasines_splitting_2023} make the weaker but less verifiable assumption that the probability of selection decays to zero slow enough. 

We also place a very mild set of conditions on the penalty function.

\begin{condition}\label{cond:penalty_union_convex}
    The penalty $\rho_\lambda(\theta)$ satisfies the following properties:
    \begin{enumerate}[(i)]
        \item The subgradients are uniformly bounded, i.e., $\sup_{\theta \in \RR^p} \sup_{\eta \in \partial\rho(\theta)} \nbr{\eta}_2 < \infty$;
        \item For any $H \in \mathbb{S}_{p}^{++}$,
            \begin{equation*}
                \cbr{H\theta + \eta : \supp\rbr{\theta} = E, \eta \in \partial \rho_\lambda(\theta)}
            \end{equation*}
            can be written as the union of finitely-many disjoint convex sets with non-empty interiors.
    \end{enumerate}
\end{condition}

\cref{cond:penalty_union_convex} is satisfied by the popular $L_1$ penalty. While a similar result is used in the foundational work of \citet{lee_exact_2016}, where a linear model with Gaussian errors is assumed, we emphasize that \cref{cond:penalty_union_convex} involves only the penalty term $\rho_\lambda(\theta)$ and places no assumptions on the distribution of the observations.

\begin{proposition}\label{prop:is_valid_penalty}
    For $\lambda \in (0, \infty)$, the lasso penalty $\rho_\lambda(\theta) = \lambda \nbr{\theta}_1$ satisfies \cref{cond:penalty_union_convex}.
\end{proposition}

Finally, we require that $\hat \Sigma_n$ is an appropriate estimator for $\Sigma_n$.
\begin{condition}\label{cond:Sigma_est}
     $\hat \Sigma_n$ satisfies $\EE\sbr{\nbr{\hat \Sigma_n - \Sigma_n}_F} = o(1)$.
\end{condition}

Recalling the definitions of $\hat Z_n^{(1)}$ and $\hat Z_n^{(2)}$ in \cref{eq:define_thins}, we are now ready to present Lemmas~\ref{lemma:selection_event} and \ref{lemma:inference_event}. 

\begin{lemma}\label{lemma:selection_event}
Under Assumption \ref{ass:existing_limits} and Conditions \ref{cond:unique_solution}--\ref{cond:Sigma_est}, the event $\left\{ \supp\rbr{\MestP} = E \right\}$ equals the event $\left\{\hat{Z}_n^{(1)} + r_n^A \rbr{\MestP} \in A_n^E \right\}$ for some $r_n^A \rbr{\MestP}$ such that $r_n^A \rbr{\MestP} = o_p(1)$, and for some $A_n^E$ that is the union of finitely many measurable disjoint convex sets.
\end{lemma}

\begin{lemma}
\label{lemma:inference_event}
Under Assumption \ref{ass:existing_limits} and Conditions \ref{cond:unique_solution}, \ref{cond:taylor_expansion}, and \ref{cond:Sigma_est}, the event \newline$\left\{ \xi^\top \MstarE \in \CI_n^{\xi, \alpha}\rbr{\MestE, \Shat } \right\}$ equals the event $\left\{ \hat{Z}_n^{(2)} + r_n^B\rbr{\MestE}  \in B_n^E  \right\}$ for some $r_n^B \rbr{\Mest}$ such that $r_n^B \rbr{\Mest} = o_p(1)$, and for some measurable convex set $B_n^E$.
\end{lemma}

Lemmas \ref{lemma:selection_event} and \ref{lemma:inference_event} formalize Insight 2 in \cref{sec:contributions} by establishing that both the selection and inference events are related to thinned versions of the score variable $\hat Z_n$. We now combine these lemmas with the conditional Berry-Esseen results in Lemma~\ref{lemma:conditional_berry_esseen} to formalize Insight 3.

\begin{theorem}[Post-selection inference via a noisy penalty]\label{thm:valid_inf_grad}
    Under Assumption \ref{ass:existing_limits} and Conditions \ref{cond:berry_esseen}--\ref{cond:Sigma_est},
    \begin{align*}
        \lim_{n \to \infty} \Pr \rbr{ \xi^\top \MstarE \in \CI_n^{\xi, \alpha}\rbr{\MestE, \Shat} \middle|\, \supp\rbr{\MestP} = E}
        \geq
        1 - \alpha.
    \end{align*}
\end{theorem}

\cref{thm:valid_inf_grad} tells us that if we  select a model $E$ using the randomized penalized M-estimator $\MestP$ in \eqref{eq:MestP}, then we can conduct inference on the parameter $\theta_n^{E,*}$ corresponding to this selected model 
by simply constructing a  standard confidence interval  of the form \eqref{eq:classical_ci} that is centered around the randomized M-estimator $\MestE$ in \eqref{eq:MestE}.

From a computational perspective, this amounts to an astoundingly simple strategy for post-selection inference for penalized M-estimators: valid inference in the model selected by the randomized penalized M-estimator $\MestP$ in \eqref{eq:MestP} simply requires software to solve the optimization problem \eqref{eq:MestP} (for both $\lambda>0$ and $\lambda=0$), e.g., using \texttt{adelie}~\citep{yang_fast_2024}.  This is in marked contrast to previous proposals from the selective inference literature such as \citet{tibshirani_uniform_2018} and \citet{huang_selective_2024}, which also consider inference on the model given by the support of $\MestP$ but require bespoke procedures for downstream inference.

\subsection{Post-selection inference by noising the samples}\label{sec:noisyresponse}

The very simple strategy for post-selection inference for penalized M-estimators suggested by Theorem~\ref{thm:valid_inf_grad}  requires the analyst to solve the randomized penalized M-estimation problem given in \eqref{eq:MestP}. In this section, we will present a corollary of Theorem~\ref{thm:valid_inf_grad}, which reveals that as long as the loss function takes an appropriate form, we can conduct post-selection inference for penalized M-estimators \emph{without} needing to solve \eqref{eq:MestP} at all. The key is to add noise to the data $Y$ rather than to the penalized M-estimator optimization problem; this mirrors (but substantially generalizes) the approaches taken by \citet{rasines_splitting_2023} and \citet{lei_discussion_2025}, and provides an asymptotic and distribution-free version of the finite-sample Gaussian  guarantees of  \cite{rasines_splitting_2023}. 

Recall from Section~\ref{subsec:m_est} our random data $Y_i \in \RR^d$. Let $\bar\Sigma_{n, i} \in \RR^{d \times d}$ be an estimate of the covariance $\Var(Y_i)$; in many settings, we will have a homogeneous covariance estimate $\bar\Sigma_{n,i}=\bar\Sigma_n$, although we write our results here more generally to accommodate arbitrarily heterogeneous covariances. We construct auxiliary noise $\bar W_{n, i} \sim \Norm(0, \bar\Sigma_{n, i})$, and define
\begin{equation}\label{eq:MestPoutcome}
    \MestPoutcome = \argmin_{\theta \in \RR^p} \cbr{\frac1n \sum_{i=1}^n \ell_\theta\rbr{Y_i + \gamma \bar W_{n, i}} + n^{-1/2} \rho_\lambda\rbr{\theta}}
\end{equation}
and
\begin{equation}\label{eq:MestEoutcome}
    \MestEoutcome :=
    \argmin_{\theta \in \RR^{\abr{E}}} \cbr{\frac1n \sum_{i=1}^n \ell_{\rbr{\theta, 0_{p-\abr{E}}}}\rbr{Y_i - \tfrac1\gamma  \bar W_{n, i}}}.
\end{equation}

We require the following condition.

\begin{condition}\label{cond:affine} The loss function 
$\ell_\theta(Y_i) = \theta^\top A_{i} Y_i + a_{\theta, i}$ for some $A_{i} \in \RR^{p \times d}$ and $a_{\theta, i} \in \RR$.
\end{condition}

Under \cref{cond:affine}, it is clear that
\begin{equation*}
    \ell_\theta(Y_i + \gamma \bar W_{n, i}) = \ell_\theta(Y_i) + \gamma \theta^\top A_i \bar W_{n, i} \quad\text{and}\quad \ell_\theta\rbr{Y_i - \tfrac1\gamma \bar W_{n, i}} = \ell_\theta(Y_i) - \tfrac1\gamma \theta^\top A_i \bar W_{n, i},
\end{equation*}
where $A_i \bar W_{n, i} \sim \Norm\rbr{0, A_i \bar\Sigma_{n, i} A_i^\top }$.
Thus, for $\hat W_n = n^{-1/2} \sum_{i=1}^n A_i \bar W_{n, i} $, where $\hat W_n$ is the noise used to estimate $\MestP$ in \cref{eq:MestP} and $\MestE$ in \cref{eq:MestE}, the objective function in \cref{eq:MestP} is equivalent to that of \cref{eq:MestPoutcome}, and \cref{eq:MestE} is equivalent to \cref{eq:MestEoutcome}. This implies that \cref{thm:valid_inf_grad} holds.

\begin{corollary}[Valid post-selection inference via a noisy response]\label{thm:valid_inf_outcome}
    Under \cref{cond:affine}, if $\hat W_n = n^{-1/2} \sum_{i=1}^n A_i \bar W_{n, i} $ then $\MestPoutcome = \MestP$ and $\MestEoutcome = \MestE$. Therefore, under Assumption \ref{ass:existing_limits} and Conditions \ref{cond:berry_esseen}--\ref{cond:Sigma_est},
    \begin{align*}
        \lim_{n \to \infty} \Pr \rbr{ \xi^\top \MstarE \in \CI_n^{\xi, \alpha}\rbr{\MestEoutcome, \Shat} \middle|\, \supp\rbr{\MestPoutcome} = E}
        \geq
        1 - \alpha.
    \end{align*}
\end{corollary}
Thus, from a computational perspective, valid post-selection inference on a model selected via penalized M-estimation requires only software to compute a penalized M-estimator, provided that the loss function is of an appropriate form. We will see that this holds for a variety of common M-estimators such as generalized linear models with canonical link functions.

\section{Application to generalized linear models}\label{sec:glms}

We now apply  our results from previous sections to the setting of a fixed-$X$ generalized linear model (GLM). Even in this very simple setting, existing selective inference tools are restrictive. For instance, consider logistic regression. Since the outcomes are Bernoulli, data thinning is not applicable~\citep{dharamshi_generalized_2025}, and the asymptotic results of \cite{lei_discussion_2025} require outcomes with continuous errors. Data fission~\citep{leiner_data_2025} can be applied in principle, using insights from \cite{neufeld_discussion_2025}, but challenges remain: for instance, (i) when the logistic regression model does not hold, the parameter targeted by data fission typically does not coincide with the true parameter of interest; and (ii) it is not clear how to target a particular allocation of Fisher information between the selection and inference steps. The asymptotic results of \cite{tian_selective_2018} are applicable, but require complex sampling approaches from a correctly-specified selective density in order to conduct inference. \citet{panigrahi_exact_2024}, \citet{huang_selective_2024}, and \citet{bakshi_selective_2024} extend these ideas to incorporate randomization.

We take our working model for $Y_i \in \RR$ to be a GLM with a canonical link function:
\begin{equation}\label{eq:exp_family}
    f_\theta (Y_i) = \exp\rbr{\frac{Y_i X_i^\top \theta - b(X_i^\top \theta)}{\alpha} + c(Y_i, \alpha)},
\end{equation}
where $X_i \in \RR^p$ are fixed regressors, $b: \RR \mapsto \RR$, 
and $\alpha$ is an over-dispersion parameter for the variance. 

Minimizing the negative log-likelihood $-\log f_\theta (Y_i)$ for the working model is equivalent to minimizing the loss
\begin{equation}\label{eq:glm_loss}
    \ell_\theta(Y_i)
    := 
     -Y_i X_i^\top \theta + b(X_i^\top \theta),
\end{equation}
which yields the score variable
\begin{equation*}
    \hat Z_n := \sqrt{n} P_n \dot \ell_{\Mstar} = \frac{1}{\sqrt{n}} \sum_{i=1}^n X_i \rbr{\dot b(X_i^\top \Mstar) - Y_i},
\end{equation*}
where $\Mstar$ was defined in \cref{eq:Mstar}. 

Our goal is valid post-selection inference with a GLM working model and an $L_1$ penalty for selection. We already established in \cref{prop:is_valid_penalty} that the $L_1$ penalty satisfies \cref{cond:penalty_union_convex}. Furthermore, examining the GLM loss function in \cref{eq:glm_loss}, it is clear that \cref{cond:affine} is satisfied with $A_{i} = -X_i$ and $a_{i, \theta} = b\rbr{X_i^\top \theta}$, and so valid inference can be obtained by adding noise to the outcomes as outlined in \cref{thm:valid_inf_outcome}. It remains to establish that the remaining conditions for \cref{thm:valid_inf_outcome} hold.

Since $X_i$ is fixed,
\begin{equation}\label{eq:glm_var}
    \Sigma_n
    := \Var\rbr{\hat Z_n}
    = \frac{1}{n} \sum_{i=1}^n \Var(Y_i) X_i X_i^\top.
\end{equation}
If $X^\top X$ is positive definite, then $\Sigma_n$ is invertible, satisfying \cref{cond:berry_esseen}(i).  \cref{cond:berry_esseen}(ii) holds under mild regularity conditions.
\cref{cond:taylor_expansion} involves a Taylor expansion of the loss around $\MstarE$ and holds under standard assumptions for the theory of M-estimation \citep[Section 5.2]{vaart_asymptotic_1998}, as established by the following result.

\begin{proposition}\label{prop:taylor_expansion_exists}
    Suppose that $\ell$ is thrice differentiable in $\theta$, and $\theta$ is restricted to a compact subset $\Theta \subset \RR^p$ over which we have uniform convergence, 
    \begin{equation*}
        \sup_{\theta \in \Theta} \sbr{P_n \ell_\theta - \EE P_n \ell_\theta} = o_p(1).
    \end{equation*}
    Furthermore, suppose that $\MstarE$ in \cref{eq:Mstar} is a unique minimizer, that $\EE P_n \ell_\theta$ has a unique minimizer $\MstarP$, and that we are in the local alternative setting where $\sqrt{n} \MstarP = O(1)$.
    Then under regularity conditions \citep{vaart_asymptotic_1998, newey_uniform_1991}, \cref{cond:taylor_expansion} holds for $H_n = H_n\rbr{\Mstar} := \EE P_n \ddot \ell_{\Mstar}$.

\end{proposition}

In the case of a GLM, the Hessian $H_n\rbr{\Mstar}$ in \cref{prop:taylor_expansion_exists} takes the form
\begin{equation}\label{eq:glm_hess}
    H_n(\theta)
    := \EE P_n \ddot \ell_{\theta}
    = P_n \ddot \ell_{\theta} = \frac{1}{n} \sum_{i=1}^n \ddot b(X_i^\top \theta) X_i X_i^\top.
\end{equation}
By construction of the working model \eqref{eq:exp_family}, $\ddot b(X_i^\top \theta) > 0$ since under the working model $\Var\rbr{Y_i} = \alpha \ddot b(X_i^\top \theta)$. Therefore, provided that
$X^\top X$ has full rank, the Hessian is positive definite. This establishes \cref{cond:unique_solution}, i.e., the uniqueness conditions of $\MestEoutcome$ and $\MestPoutcome$, along with the uniqueness conditions of \cref{prop:taylor_expansion_exists} (see \citet{tibshirani_lasso_2013} for a more detailed analysis of the settings under which $L_1$-penalized problems have unique solutions).

Next, we turn to \cref{cond:Sigma_est}, which involves estimating $\Sigma_n$ defined in \cref{eq:glm_var} using an estimator $\bar{\Sigma}_{n,i}$ of $\Var(Y_i)$. One approach, as taken by \citet{huang_selective_2024}, is to assume that the the working model \eqref{eq:exp_family} holds for some particular $\MstarP$. Then  $\Var\rbr{Y_i} = \alpha \ddot b(X_i^\top \MstarP)$ and $E(Y_i) = \dot b(X_i^\top \MstarP)$. Define the over-dispersion estimator 
\begin{equation}\label{eq:overdispersion}
    \hat \alpha(\theta) = \frac{1}{n - p} \sum_{i=1}^n \frac{\rbr{Y_i - \dot b(X_i^\top \theta)}^2}{\ddot b(X_i^\top \theta)}
\end{equation}
Assuming that the working model holds, classic GLM results provide a consistent estimator $\tilde \theta_n$ of $\MstarP$, which in turn provides a consistent estimator $\hat \alpha\rbr{\tilde \theta_n} \ddot b\rbr{X_i^\top \tilde \theta_n} $ of $\Var(Y_i)$. Under stronger regularity conditions, i.e., uniform integrability or convergence in distribution, consistency can be strengthened to $L_1$ consistency required for \cref{cond:Sigma_est} to hold \citep[Ch. 2.5]{vaart_asymptotic_1998}. Alternatively, nonparametric variance estimation approaches exist, e.g. using kernel-based regression methods \citep{hall_variance_1989, fan_efficient_1998}.

Finally, we turn to \cref{cond:sandwich_est}, which requires a variance estimator $\Shat$ satisfying
\begin{equation}\label{eq:outcome_sandwich}
    \lim_{n \to \infty} \Pr\rbr{\xi^\top \MstarE \in \CI_n^{\xi, \alpha}\rbr{\MestEoutcome, \Shat }} \geq 1-\alpha.
\end{equation}
Since $\MestEoutcome$ is the M-estimator in \cref{eq:MestEoutcome}, where $E$ defines a fixed model, \cref{eq:outcome_sandwich} is satisfied by the classic model-robust sandwich estimator $\Shat := \rbr{H_{n, E, E}(\MestEoutcome)}^{-1} \tilde \Sigma_{n, E, E} \rbr{H_{n, E, E}(\MestEoutcome)}^{-1}$~\citep{white_heteroskedasticity-consistent_1980}, where $H_{n, E, E}(\theta)$ is a submatrix of the Hessian defined in \cref{eq:glm_hess} and $\tilde \Sigma_{n, E, E}$ is a diagonal matrix of squared residuals $\rbr{\rbr{Y_i - \tfrac{1}{\gamma} \bar W_{n, i} } - \dot b(X_{i, E}^\top \MestEoutcome)}^2$.

This establishes all of the conditions necessary for application of \cref{thm:valid_inf_outcome}, whose full procedure is outlined in Algorithm \ref{algo:valid_inf_outcome}. A similar procedure implementing \cref{thm:valid_inf_grad} is given in \cref{sec:app_grad_algo}.

\begin{mybox}[label=algo:valid_inf_outcome]{implementing \cref{thm:valid_inf_outcome} for $L_1$-penalized GLMs}
\textbf{Input:} Data $\cbr{Y_i, X_i}$, penalty $\lambda > 0$, information split $\gamma > 0$.
\begin{enumerate}
    \item ~$\tilde\theta_n \leftarrow
    \argmin_{\theta \in \RR^{p}} P_n \ell_{\theta}$.
    \item ~Sample noise $\bar W_{n, i} \leftarrow \Norm\rbr{0,  \hat\alpha(\tilde\theta_n) \ddot b \rbr{X_i^\top \tilde\theta_n} }$.
    \item ~Estimate $\MestPoutcome$ as in \cref{eq:MestPoutcome} and define $E := \supp\rbr{\MestPoutcome}$.
    \item ~Estimate $\MestEoutcome$ as in \cref{eq:MestEoutcome}.
    \item ~Compute sandwich variance estimator using outcomes $\cbr{Y_i - \tfrac{1}{\gamma} \bar W_{n, i}}_{i=1}^n$:\footnotesize
    \begin{equation*}
        \Shat \leftarrow \rbr{H_{n, E, E}\rbr{\MestEoutcome}}^{-1} \sbr{\frac1n \sum_{i=1}^n X_{i, E} \rbr{Y_i - \tfrac{1}{\gamma} \bar W_{n, i} - \dot b(X_{i, E}^\top \MestEoutcome)}^2 X_{i, E}^\top}
    \rbr{H_{n, E, E}\rbr{\MestEoutcome}}^{-1}.
    \end{equation*}
\end{enumerate}
\textbf{Output:} A confidence interval centered at $\MestEoutcome$:
\begin{equation*}
    \CI^{\xi, \alpha}_n\rbr{\MestEoutcome, \Shat} = \sbr{\xi^\top \MestEoutcome \pm n^{-1/2} z_{1-\alpha/2}\rbr{ \xi^\top \Shat \xi}^{1/2}}.
\end{equation*}
\end{mybox}

\section{Experimental results}\label{sec:applications}

We now empirically study the procedure outlined in Algorithm \ref{algo:valid_inf_outcome}, first in simulation and then on the \emph{Teenage Friends and Lifestyle Study} dataset. All code is provided at \href{https://github.com/rflperry/m\_estimation\_SI}{https://github.com/rflperry/m\_estimation\_SI}.

\subsection{Simulation studies: comparisons to data thinning}

In the following simulations, we generated outcomes from a generalized linear model with canonical link function, i.e., $ \EE[Y_i] = g^{-1}\rbr{X_i^\top \theta_n^{*}}$.
The covariates $X_i \in \RR^{40}$ are normally distributed with variance one and an equicorrelation of $0.3$. The true mean function $\theta_n^{*}$ is a sparse vector with ten nonzero entries. Selection is performed using the lasso penalty $\rho_\lambda(\theta) := \lambda \nbr{\theta}_1$, where $\lambda \propto \sqrt{(1 + \gamma^2) \log(p)} \, \mathrm{sd}(Y_1, \dots, Y_n)$; recall from \cref{eq:MestPoutcome} that this penalty is then further scaled by $n^{-1/2}$. We set $\gamma = 1$, which effectively allocates half of the information in the data for selection and half for inference.

Our first simulation investigates the setting of a linear mean and homoskedastic Gaussian noise with an unknown variance, i.e., $Y_i = X_i^\top \theta_n^{*} + \varepsilon_i$. We apply Algorithm \ref{algo:valid_inf_outcome} to conduct inference on the model selected using $L_1$-penalized ordinary least squares. In this setting, score thinning closely aligns with the Gaussian data thinning proposals of \citet{rasines_splitting_2023} (with a different variance estimator) and \citet{lei_discussion_2025}. We compare score thinning to the randomized conditional selective inference approaches of \citet{panigrahi_exact_2024} and \citet{huang_selective_2024}, both of which have the \emph{exact} same selection procedure as score thinning but rely on bespoke conditional likelihoods for inference. Lastly, we include comparisons to sample splitting (which uses half of the samples for selection and half for inference, in order to match the information split of $\gamma = 1$) and the classical approach (which reuses the data for selection and inference).

In the left-hand panel of \cref{fig:linear}, we display the average coverage of $90\%$ confidence intervals for all parameters in each selected model; a method that controls the conditional coverage in the sense of \cref{eq:conditional_coverage} will attain $90\%$ coverage. The classical method (which reuses the data for selection and inference) fails to achieve the nominal coverage, and thus we exclude it from the remaining panels. All other approaches attain 90\% coverage, except that of \citet{huang_selective_2024}, which approximates a selective MLE and has been shown to undercover in practice \citep{panigrahi_exact_2024}. In contrast, \citet{panigrahi_exact_2024} employ an \emph{exact} approach that attains the nominal coverage. In the middle panel of \cref{fig:linear}, we assess the power of each method by comparing the distribution of average confidence interval widths. Score thinning yields wider intervals on average than the randomized conditional inference approaches of \citet{panigrahi_exact_2024} and \citet{huang_selective_2024}, which is expected since the conditional approaches are able to use more information for inference and since \citet{huang_selective_2024} does not attain the nominal coverage. Score thinning yields narrower intervals than sample splitting, as suggested by Proposition 1 of \citet{rasines_splitting_2023}. The poor performance of sample splitting may also be due to the fact that sample splitting can lead to a poorly-conditioned design matrix when $n$ is small. Lastly, we assess the model selection accuracy in the right-hand panel of \cref{fig:linear} by comparing the average false discovery rate (FDR) of each selective inference method. The FDR of a model $E$ is given by $\abr{E \setminus \supp\rbr{\theta_n^*}} / |E|$, i.e., the fraction of incorrectly selected features relative to the total number of selected features. As mentioned earlier, score thinning and the randomized conditional approaches of \citet{huang_selective_2024} and \citet{panigrahi_exact_2024} have identical selection events and therefore have identical FDRs, which are comparable to that of sample splitting; all methods select models of a similar size $|E|$.

\begin{figure}[!htb]
    \centering
    \includegraphics[width=1.0\linewidth]{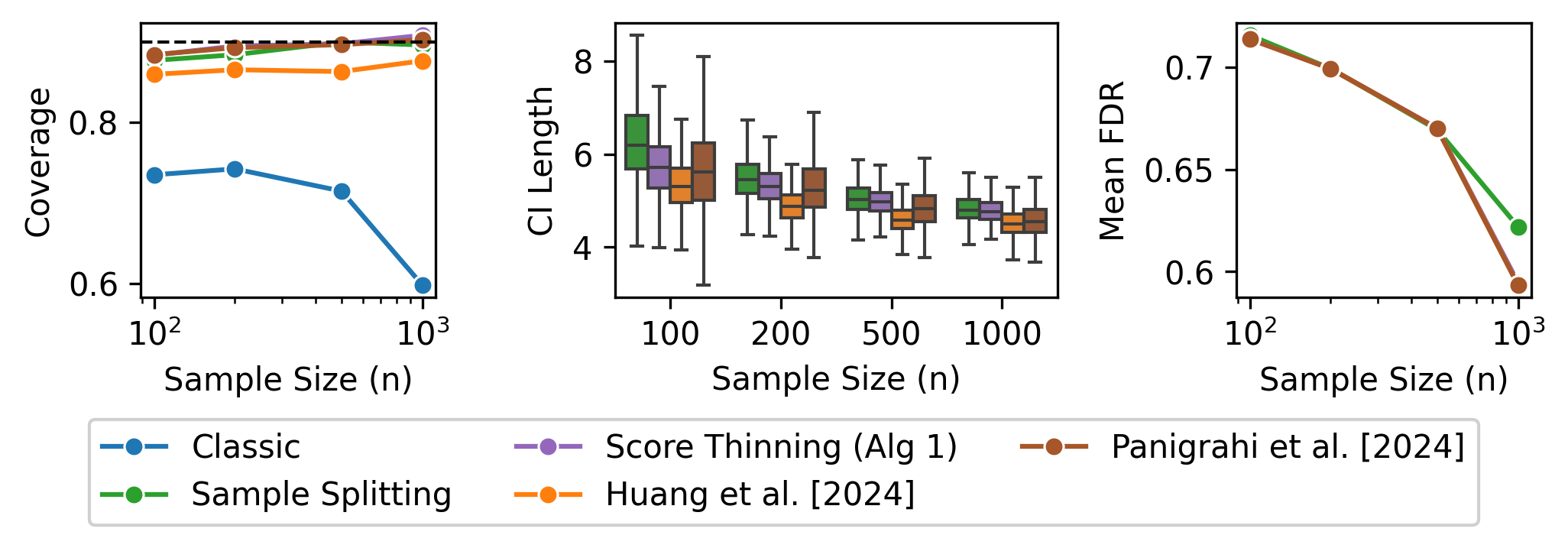}
    \caption{Outcomes are generated from a linear model with i.i.d. Gaussian noise whose variance is unknown. Inference is conducted on the linear model selected by $L_1$-penalized ordinary least squares. Results are aggregated across $1000$ simulated datasets. Since the classical approach (which reuses the data for selection and inference) fails to attain valid coverage, we omit it from the center and right-hand panels.}
    \label{fig:linear}
\end{figure}

In our second simulation, outcomes are generated from a Poisson regression model, with and without Gamma-distributed overdisperion. With overdispersion, the outcomes are marginally negative binomial (NB). Inference is conducted on the model selected using $L_1$-penalized Poisson regression. Since the outcomes are Poisson or NB, data thinning can be applied to obtain datasets for selection and inference that are \emph{exactly} independent~\citep{neufeld_data_2024}. Score thinning also applies. In the top row of \cref{fig:poisson} (in the absence of overdispersion), we see that all of the selective inference methods attain nominal coverage and perform similarly. However, in the bottom row (in the presence of overdispersion), we see that Poisson thinning fails to attain nominal coverage; this is because applying Poisson thinning to NB data does not yield independent datasets. NB thinning and our score thinning proposal have comparable performances; however, while the former requires correct specification of the data-generating distribution, the latter requires much weaker assumptions.

\begin{figure}[!htb]
    \centering
    \includegraphics[width=1.0\linewidth]{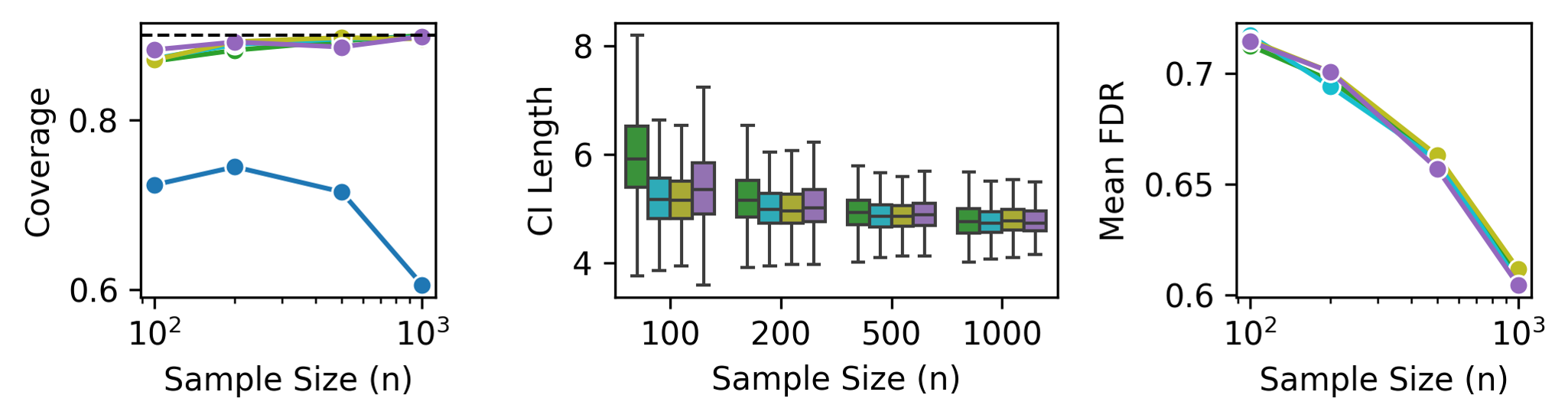}
    \includegraphics[width=1.0\linewidth]{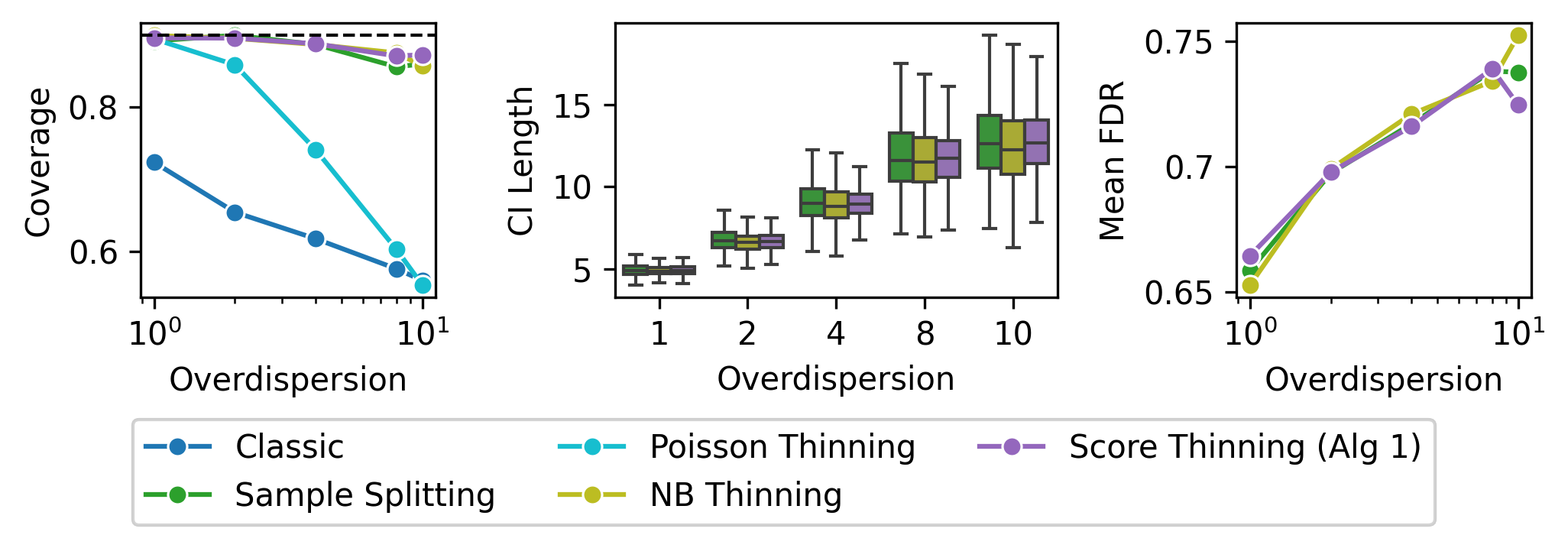}
    \caption{We simulate outcomes from a Poisson regression model with and without overdispersion. Inference is conducted on the model selected using $L_1$-penalized Poisson regression. Confidence intervals are computed for all coefficients in the selected model. Results are aggregated across $1000$ simulated datasets. Since the classical approach (which reuses the data for selection and inference) fails to attain nominal coverage, we omit it from the center and right-hand panels. \textbf{(Top):} In the absence of overdispersion, as we vary the sample size, all thinning methods perform comparably. \textbf{(Bottom):} We vary the overdispersion of the variance relative to the mean; this induces a negative binomial outcome distribution. Since Poisson thinning fails to attain valid coverage, we omit it from the center and right-hand panels.}
    \label{fig:poisson}
\end{figure}

\subsection{Simulation studies: comparisons beyond data thinning}

We now move to settings in which data thinning \emph{cannot be applied} to obtain \emph{exactly} independent datasets for selection and inference. We first consider outcomes generated from a logistic regression model. Inference is conducted on the model selected using $L_1$-penalized logistic regression. Since the outcomes are Bernoulli, it is not possible to thin them into independent datasets~\citep{dharamshi_generalized_2025}, although score thinning still applies. The left-hand panel of \cref{fig:logistic} shows that the classical method (which reuses the data for selection and inference) fails to attain valid coverage. The coverage of \citet{huang_selective_2024} once again falls below the nominal level. The exact approach of \citet{panigrahi_exact_2024} is not displayed as it is only applicable to the linear Gaussian setting. In principle, the exact approach of \citet{bakshi_selective_2024} could be applied, although they do not provide software for their extension to logistic regression. In the center panel of \cref{fig:logistic}, the randomized conditional method of \citet{huang_selective_2024} has the narrowest intervals, possibly because they fail to attain the nominal coverage. Sample splitting yields the widest intervals, again likely due to the ill-conditioning of the design matrix. In the right-hand panel of \cref{fig:logistic}, the FDRs are comparable.

We note that existing software for (penalized) logistic regression, e.g, \verb=glmnet= in \verb=R= \citep{friedman_regularization_2010}, typically requires that the inputted outcome values are binary or at least non-negative. Because the noisy outcomes in \cref{eq:MestPoutcome} and \cref{eq:MestEoutcome} are non-binary and contain negative values, we cannot use \verb=glmnet= out of the box to solve \cref{eq:MestPoutcome} or \cref{eq:MestEoutcome}. Instead, our software implementation makes use of the \verb=Python= package \verb=RegReg=~\citep{regreg}.

\begin{figure}[!htb]
    \centering
    \includegraphics[width=1.0\linewidth]{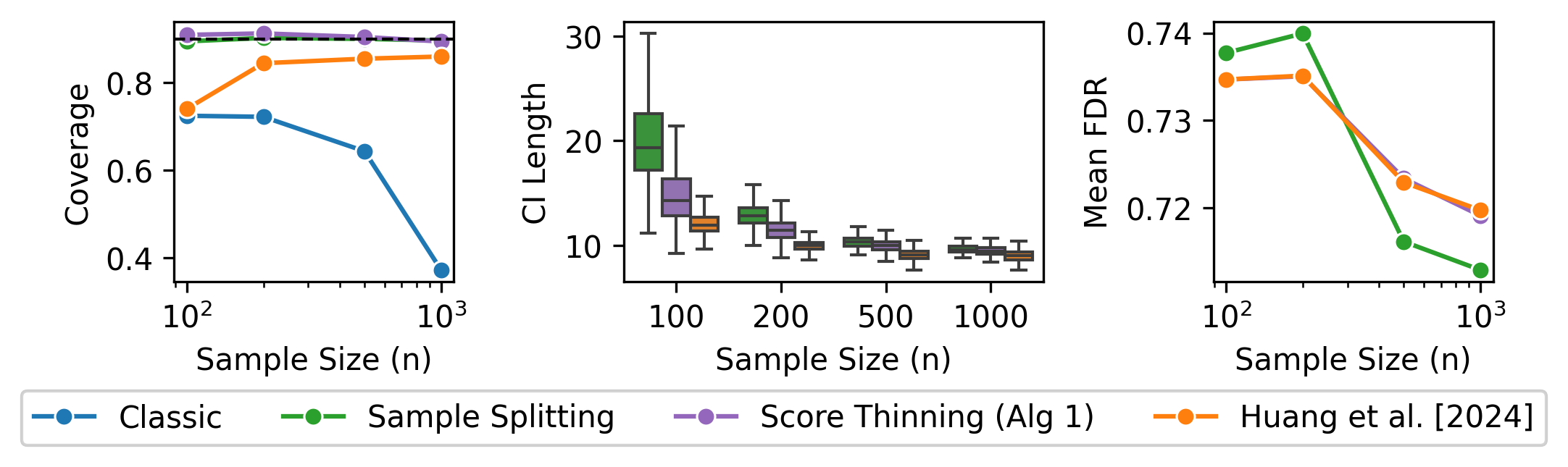}
    \caption{We simulate outcomes from a logistic regression model. Inference is conducted on the model selected using $L_1$-penalized logistic regression. Confidence intervals are computed for all coefficients in the selected model. Results are aggregated across $1000$ simulated datasets. Since the classical approach (which reuses the data for selection and inference) fails to attain valid coverage, we omit it from the center and right-hand panels.}
    \label{fig:logistic}
\end{figure}

In our final simulation, we consider a more involved setting for which no existing methods are available in the selective inference literature: cluster-correlated data. We simulate observations whose mean is linear in the covariates $X_{i, j}$ but which have a cluster-specific intercept, i.e., $Y_{i, j} = X_{i, j}^\top \theta_n^{*} + \nu_i + \varepsilon_{i, j}$ where $\nu_i \sim \Norm(0, 0.5)$,  $\varepsilon_{i, j}$ are i.i.d. symmetric Laplace random variables with variance $0.5$, $j = 1, \dots, 10$, and $i = 1, \dots, n/10$; at present, no method is available to thin a Laplace random variable into independent components.
Thus, $n$ observations are grouped into equicorrelated clusters of size ten. Algorithm \ref{algo:valid_inf_outcome} is not immediately applicable to this setting since the working model in \cref{eq:exp_family} assumed independent outcomes rather than cluster-correlated outcomes. We therefore apply Algorithm \ref{algo:valid_inf_outcome} with minor modifications: in Step 2 we estimate the marginal variance and equicorrelation parameter, and in Step 5 we use a sandwich variance estimator with an equicorrelated working model. \cref{fig:clustered} shows that the classical approach (which reuses the data for selection and inference) again fails to achieve the nominal coverage, while score thinning attains narrower confidence intervals than sample splitting. We note that sample splitting is performed on the cluster level in order for the independence of the selection and inference steps to be preserved. We do not compare to a randomized conditional approach as no existing proposal can accommodate this setting.

\begin{figure}[!htb]
    \centering
    \includegraphics[width=1.0\linewidth]{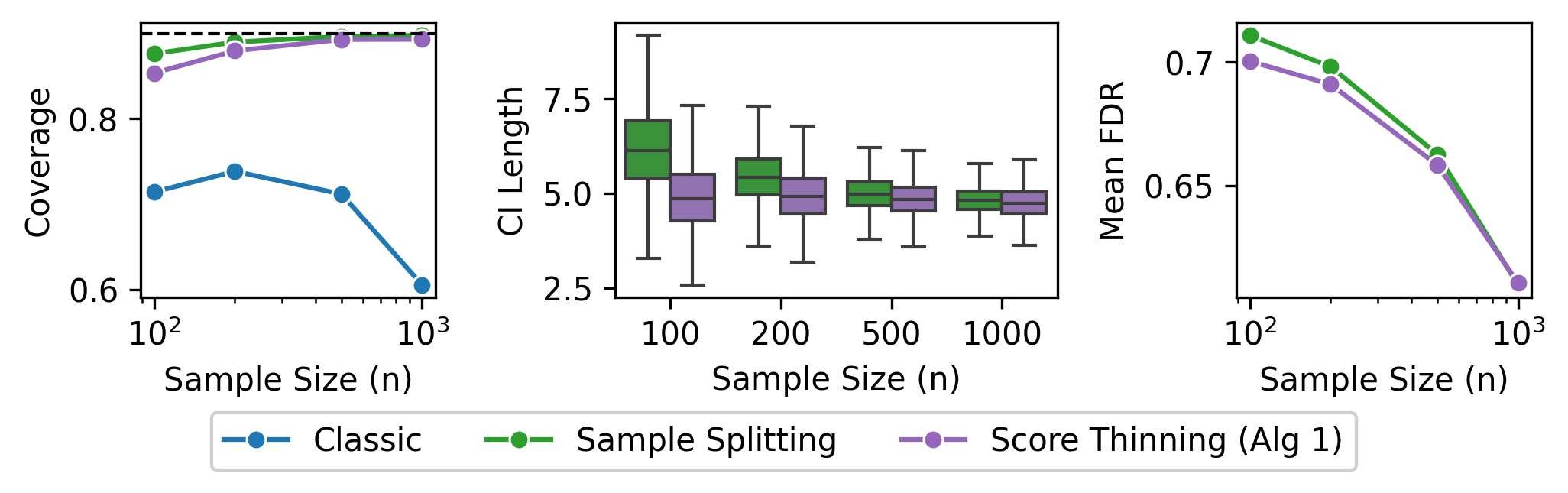}
    \caption{Outcomes are generated from a linear model with equicorrelated non-Gaussian noise. Inference is conducted in the linear model selected by $L_1$-penalized linear regression. Results are aggregated across $5000$ simulated datasets. Since the classical approach (which reuses the data for selection and inference) fails to attain valid coverage, we omit it from the center and right-hand panels.}
    \label{fig:clustered}
\end{figure}

\subsection{Application to principal components network regression}

The \emph{Teenage Friends and Lifestyle Study} dataset contains the survey responses of secondary school students in Glasgow \citep{michell_peer_1996}. The students were surveyed three times between 1995 and 1997, and asked various questions about their lifestyle including their smoking usage and who their friends were. Thus, the observations lie on an observed friendship network.

Using the first wave of survey responses, \citet{hayes_estimating_2025} estimate the association between the sex of the student and their odds of smoking. Without accounting for the network, a logistic regression model accounting for age and church attendance estimates the odds ratio of being a smoker when male rather than female to be $0.30$ with a $95\%$ confidence interval of $(0.11, 0.81)$. However, we can see from the left and middle panels of \cref{fig:glasgow} that there are groups of friends who smoke.

To determine whether the association between sex and smoking persists after accounting for the friendship network, \citet{hayes_estimating_2025} further include the singular vectors from the singular value decomposition of the network's adjacency matrix as regression covariates. However, in doing so they face two uncertain modeling choices. First, since the graph is directed, the adjacency matrix is non-symmetric and thus the left and right singular vectors are different. Specifically, the former relate to outgoing edges in the network while the latter relate to incoming edges. Second, the number of singular vectors $r$ to include in the regression is unknown, and implies an assumption on the rank of the network. The authors arbitrarily choose to include only the right singular vectors in the regression, and then present results for $r =1, \dots, 25$.

To adaptively choose a model, we employ score thinning in Algorithm \ref{algo:valid_inf_outcome} (using an $L_1$-penalized logistic regression model with $\gamma = 1$ and the same $\lambda$ as in the simulations). The $n \times p$ design matrix $X$ consists of $n=152$ subjects and $p=56$ features, which consist of an intercept, age, sex, church attendance (four categories), and all of the first $25$ left and right singular vectors. Since the singular vectors used to construct $X$ are based on the entire network, sample splitting is not applicable. The noisy outcomes used for selection and inference are plotted in the right-hand panel of \cref{fig:glasgow}.

The singular vectors selected via $L_1$-penalized logistic regression using the noisy responses include a mix of non-consecutive left and right singular vectors. The estimated odds ratio of being a smoker when male rather than female is $0.43$ with a $95\%$ confidence interval of $(0.06, 3.1)$. \citet{hayes_estimating_2025} similarly found non-significant estimates, in line with the hypothesis that the apparent association between sex and smoking is the spurious result of the network structure.

\begin{figure}[!htb]
    \centering
    \includegraphics[width=0.38\linewidth]{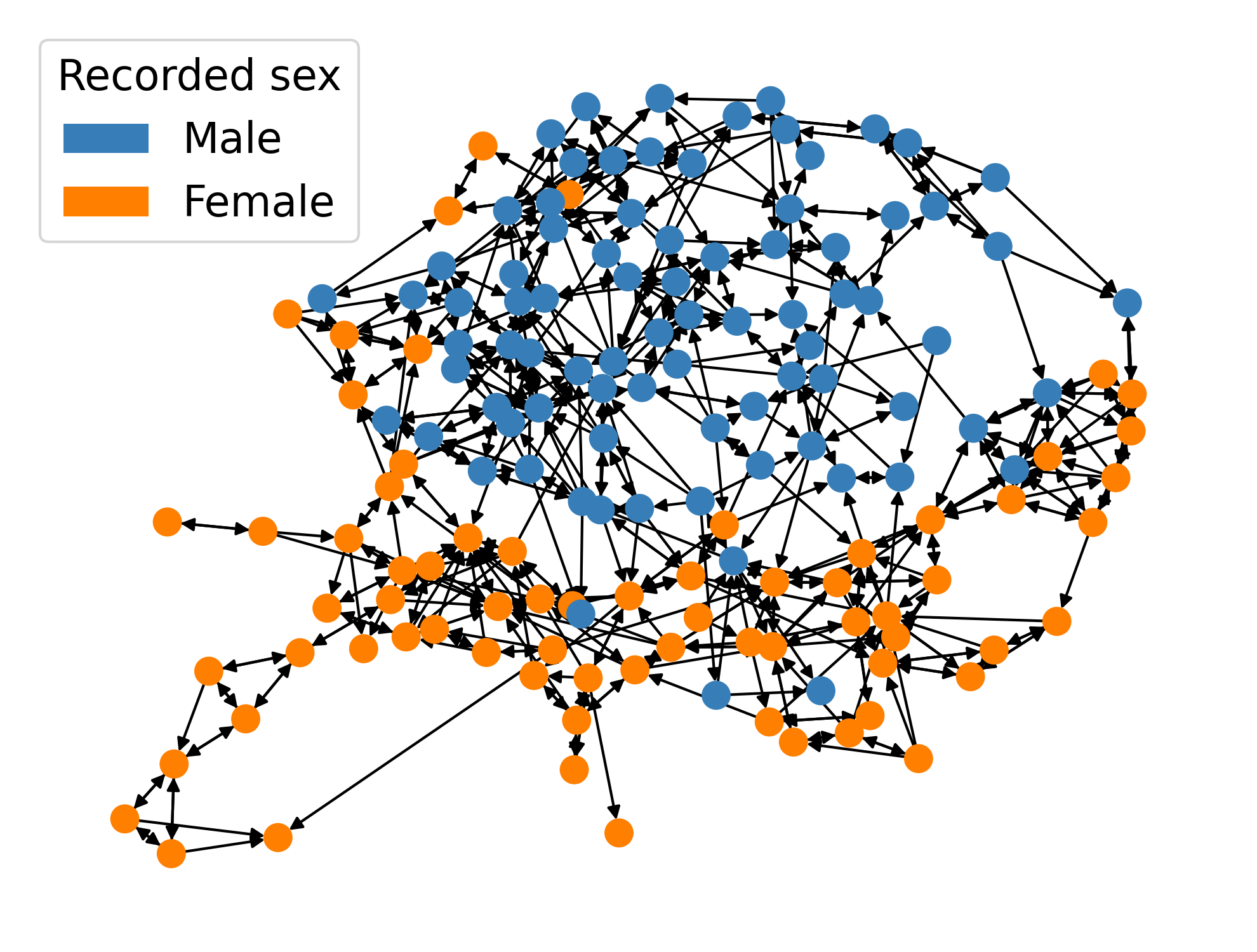}
    \includegraphics[width=0.38\linewidth]{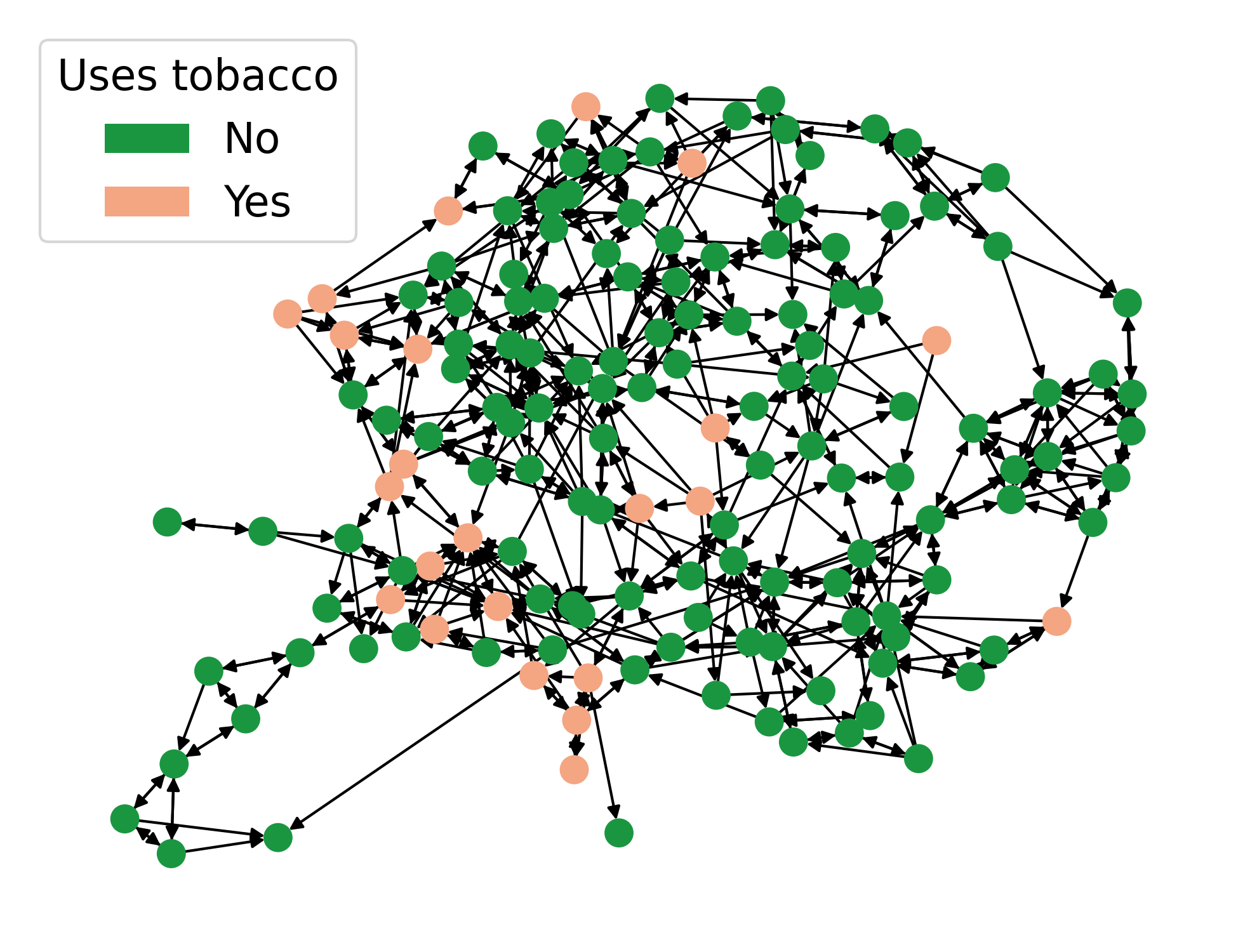}
    \includegraphics[width=0.22\linewidth]{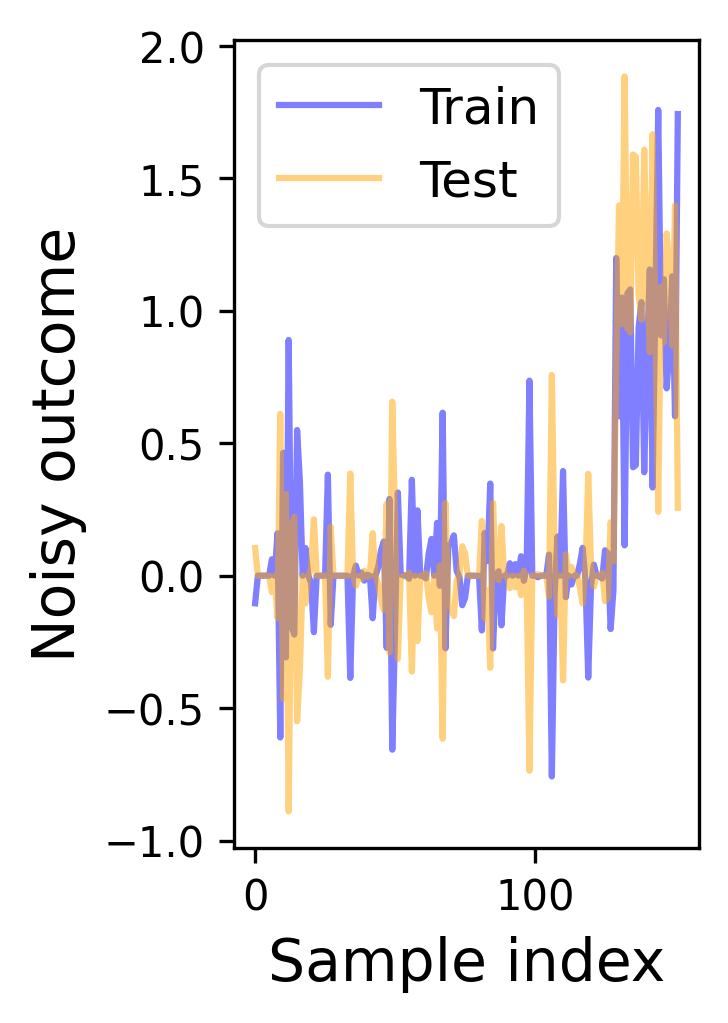}
    \caption{(Left): the friendship network of students, colored by recorded sex. (Middle): the friendship network of students, colored by their reported tobacco usage $Y$. (Right): the noisy outcomes constructed by adding and subtracting Gaussian noise to the outcomes $Y$, i.e., $Y_i + \bar W_{n, i}$ and $Y_i - \bar W_{n, i}$. These sets are not independent, but selection and inference using them are asymptotically independent.}
    \label{fig:glasgow}
\end{figure}

\section{Discussion}\label{sec:discussion}

Existing approaches to selective inference trade off statistical efficiency, ease of implementation, flexibility of the selection procedure, and distributional assumptions. Fully conditional approaches maximize the statistical efficiency in the sense of \citet{fithian_optimal_2017} by conditioning on the minimal amount of information; however, these methods require bespoke derivations and computational implementations for each selection procedure~\citep{lee_exact_2016, tian_selective_2018, panigrahi_exact_2024, huang_selective_2024}. On the other hand, applying data thinning to the response to obtain exactly independent datasets for selection and inference can accommodate arbitrary selection procedures and is easy to implement; yet, data thinning is less efficient and requires strong distributional assumptions~\citep{rasines_splitting_2023, neufeld_data_2024, dharamshi_generalized_2025}.

We have introduced score thinning, which is situated in between these two classes of approaches. Our results suggest that adding and subtracting Gaussian noise --- that is, conducting ``Gaussian data thinning'' on non-Gaussian data --- may be somewhat robust to model misspecification in penalized M-estimation post-selection inference tasks, provided that the sample size is sufficiently large.

Furthermore, our work demonstrates an equivalence between adding and subtracting noise from the outcome, as in Algorithm \ref{algo:valid_inf_outcome}, and adding and subtracting noise from the gradient, as in Algorithm \ref{algo:valid_inf_grad} and other selective inference papers \citep{tian_selective_2018, panigrahi_exact_2024, huang_selective_2024}. While the latter may be more widely applicable to penalized M-estimation problems, the former is often more practical as it avoids the need for bespoke software.

Our work suggests many new questions. First, our theoretical results are limited to the setting where $p$ is fixed, and yet the setting where $p$ is growing is also of substantial interest. Second, while \cref{cond:penalty_union_convex} is satisfied by the $L_1$ penalty, it remains to show that it holds for other popular penalties such as the group lasso. Lastly, the idea of thinning functions of the data, rather than the data itself, is a widely applicable idea which may allow for further development of strategies inspired by data thinning in the absence of distributional assumptions.

\section{Acknowledgments}
The authors thank the anonymous reviewers for their valuable suggestions. This work is supported in part by: an Amazon AI Fellowship to Ronan Perry; NSF CAREER Award DMS-2337882 to Snigdha Panigrahi; and NSF DMS-2322920, NSF DMS-2514344, NIH 5P30DA048736, and ONR N00014-23-1-2589 to Daniela Witten.

\newpage
\bibliographystyle{plainnat}
\bibliography{references}

\appendix

\setcounter{equation}{0}
\renewcommand{\theequation}{\thesection.\arabic{equation}}
\setcounter{theorem}{0}
\renewcommand{\thetheorem}{\thesection.\arabic{theorem}}
\setcounter{lemma}{0}
\renewcommand{\thelemma}{\thesection.\arabic{lemma}}

\section{Supporting theory and their proofs}

We state and prove technical results necessary for proving results in the main text.

\subsection{An extended Berry-Esseen bound}

We first recall the multivariate Berry-Esseen Theorem.
    \begin{theorem}[Berry-Esseen~\citep{raic_multivariate_2019}]\label{thm:berry_esseen}
        Let $U_i$ be independent mean zero random variables in $\RR^p$ and let $V_n := \sum_{i=1}^n U_i$ such that $\Var(V_n) = \sum_{i=1}^n \Var(U_i) = I_p$. Define $V \sim \Norm(0, I_p)$. For any measurable and convex set $A \subseteq \RR^p$,
        \begin{equation*}
            \abr{\Pr\rbr{V_n \in A} - \Pr\rbr{V \in A}} \leq \rbr{42 p^{1/4} + 16} \sum_{i=1}^n \EE\nbr{U_i}_2^3.
        \end{equation*}
    \end{theorem}
    Equivalently, this result says that
    \begin{equation*}
        \abr{\Pr\rbr{V_n \in A} - \Pr\rbr{V \in A}} = O\rbr{p^{1/4} \sum_{i=1}^n \EE\nbr{U_i}_2^3}
    \end{equation*}
    for any set $A \subseteq \RR^p$ that is convex and measurable.

    The following lemma follows from \cref{thm:berry_esseen}.
    \begin{lemma}\label{lemma:extended_be}
        Let $U_i$ be independent random variables in $\RR^p$ and assume that $\mu_{n, i} := \EE[U_i]$ is finite. Define $\hat V_n := n^{-1/2} \sum_{i=1}^n U_i$ and assume $\Sigma_n := \Var(\hat V_n) = \frac1n \sum_{i=1}^n \Var(U_i)$ is invertible. Define $\mu_n := \tfrac1n \sum_{i=1}^n \mu_{n, i}$ and let $V_n \sim \Norm(\sqrt{n} \mu_n, \Sigma_n)$. Then for any set $A \subseteq \RR^p$ which can be expressed as a union of $K$ measurable almost surely disjoint and convex sets, it holds that
        \begin{equation*}
            \abr{\Pr\rbr{\hat V_n \in A} - \Pr\rbr{V_n \in A}} = O\rbr{K p^{1/4} n^{-3/2} \eigmin{\Sigma_n}^{-3/2} \sum_{i=1}^n \EE \nbr{U_i - \mu_{n, i}}_2^3}.
        \end{equation*}
    \end{lemma}
    \begin{proof}\,

    Define $\tilde A := \cbr{\Sigma_n^{-1/2}(a - \sqrt{n}\mu_n): a \in A}$ (omitting a dependence on $n$ for brevity), and note that
    \begin{align}
        &\abr{\Pr\rbr{\hat V_n \in A} - \Pr\rbr{V_n \in A}}\nonumber\\
        &= \abr{\Pr\rbr{\Sigma_n^{-1/2} \rbr{\hat V_n - \sqrt{n}\mu_n} \in \tilde A} - \Pr\rbr{\Sigma_n^{-1/2} \rbr{V_n - \sqrt{n}\mu_n} \in \tilde A}}.\label{eq:extended_be_proof}
    \end{align}
    The term
    \begin{equation*}
        \Sigma^{-1/2}_n \rbr{\hat V_n - \sqrt{n} \mu_n} = \sum_{i=1}^n n^{-1/2} \Sigma^{-1/2}_n \rbr{U_i - \mu_{n, i}}
    \end{equation*}
    is the sum of mean-zero and independent observations $n^{-1/2} \Sigma^{-1/2}_n \rbr{U_i - \mu_{n, i}}$. Furthermore, $\Var\rbr{\Sigma^{-1/2}_n \rbr{\hat V_n - \sqrt{n} \mu_n}} = I_p$ and $\Sigma^{-1/2}_n \rbr{V_n - \sqrt{n} \mu_n} \sim \Norm(0, I_p)$.

    By assumption, $\tilde A$ remains a union of at most $K$ almost surely disjoint, measurable, and convex sets, since the transformation induced by $\Sigma^{-1/2}_n$ and $\mu_n$ preserves these properties.
    So, we can write $\tilde A = \bigcup_{k=1}^K \tilde A_k$ where $\cbr{\tilde A_k}_{k=1}^K$ are almost surely disjoint, measurable, and convex. It follows that
    \begin{align*}
        &\abr{\Pr\rbr{\Sigma^{-1/2}_n \rbr{\hat V_n - \sqrt{n}\mu_n} \in \tilde A} - \Pr\rbr{\Sigma^{-1/2}_n \rbr{ V_n - \sqrt{n}\mu_n}  \in \tilde A}}\\
        &= \abr{
            \sum_{k=1}^{K}\sbr{\Pr\rbr{\Sigma^{-1/2}_n \rbr{\hat V_n - \sqrt{n}\mu_n}  \in \tilde A_{k}} - \Pr\rbr{\Sigma^{-1/2}_n \rbr{V_n - \sqrt{n}\mu_n} \in \tilde A_{k}}}
        }\\
        &\leq \sum_{k=1}^{K}\abr{
            \Pr\rbr{\Sigma^{-1/2}_n \rbr{\hat V_n - \sqrt{n}\mu_n}  \in \tilde A_{k}} - \Pr\rbr{\Sigma^{-1/2}_n \rbr{ V_n - \sqrt{n}\mu_n}  \in \tilde A_{k}}
        }\\
        &\leq \sum_{k=1}^{K} O\rbr{ p^{1/4}  \sum_{i=1}^n \EE \nbr{n^{-1/2}\Sigma^{-1/2}_n \rbr{U_i - \mu_{n, i}}}_2^3}\\
        &=O\rbr{ K p^{1/4} n^{-3/2} \eigmin{\Sigma_n}^{-3/2} \sum_{i=1}^n \EE \nbr{U_i - \mu_{n, i}}_2^3}.
    \end{align*} 
    Combining this with \cref{eq:extended_be_proof}, we obtain our bound. 
\end{proof}

\subsection{Helper lemma for \cref{lemma:thinning_core_stat}}

We begin with a helper lemma that will be used in the proof of \cref{lemma:thinning_core_stat}.

\begin{lemma}\label{lemma:thinning_core_stat_known_var}
    Suppose that $W_n \sim \Norm(0, \Sigma_n)$ for $\Sigma_n$ defined in \cref{cond:berry_esseen}, and define
    \begin{align*}
        \tilde Z_n^{(1)} := \hat Z_n + \gamma W_n\quad\text{and}\quad
        \tilde Z_n^{(2)} := \hat Z_n - \frac1\gamma W_n.
    \end{align*}
    Under \cref{cond:berry_esseen}, for any set $C$ which is the union of at most $K$ measurable convex sets,
    \begin{align}\label{eq:thinning_core_stat_known_var}
        \abr{
        \Pr\rbr{\rbr{\tilde Z_n^{(1)}, \tilde Z_n^{(2)}} \in C} - \Pr\rbr{\rbr{Z_n^{(1)}, Z_n^{(2)}} \in C}
        } = O\rbr{2^K p^{1/4} n^{-1/2} \eigmin{\Sigma_n}^{-3/2}}.
    \end{align}
\end{lemma}
\begin{proof}
We wish to apply \cref{lemma:extended_be}, and thus require that $\rbr{\tilde Z_n^{(1)}, \tilde Z_n^{(2)}}$ are equal in distribution to the sum of independent random variables with a specified covariance matrix. Let $W_{n, i} \iidsim \Norm(0, \Sigma_n)$. First, note that
\begin{align*}
    \tilde Z_n^{(1, 2)}
    := \begin{pmatrix}
        \tilde Z_n^{(1)}\\
        \tilde Z_n^{(2)}
    \end{pmatrix}
    &=
    n^{-1/2}\sum_{i=1}^n\begin{pmatrix}
        \dot \ell_{\Mstar}(X_i, Y_i)\\
        \dot \ell_{\Mstar}(X_i, Y_i)
    \end{pmatrix} 
    +\begin{pmatrix}
        \gamma W_n\\
        - \gamma^{-1} W_n
    \end{pmatrix}\\
    &\overset{d}{=}
    n^{-1/2}\sum_{i=1}^n\begin{pmatrix}
        \dot \ell_{\Mstar}(X_i, Y_i) + \gamma W_{n, i}\\
        \dot \ell_{\Mstar}(X_i, Y_i) - \gamma^{-1} W_{n, i}
    \end{pmatrix}
\end{align*}
is clearly the sum of independent random variables; $\tilde Z_n^{(1, 2)}$ has mean $\rbr{\sqrt{n} \mustar, \sqrt{n}\mustar}$. Second, $\tilde Z_n^{(1, 2)}$ has covariance matrix
\begin{align}
    \Var\rbr{
    \begin{pmatrix}
        \tilde Z_n^{(1)}\\
        \tilde Z_n^{(2)}
    \end{pmatrix}
    }
    &=
    \Var\rbr{
    \begin{pmatrix}
        \tilde Z_n\\
        \tilde Z_n
    \end{pmatrix} 
    + \begin{pmatrix}
        \gamma W_n\\
        - \gamma^{-1} W_n
    \end{pmatrix}
    }\nonumber\\
    &= \begin{bmatrix}
        \Sigma_n & \Sigma_n\\
        \Sigma_n & \Sigma_n
    \end{bmatrix}
    + \begin{bmatrix}
        \gamma^2\Sigma_n & -\Sigma_n\\
        -\Sigma_n & \gamma^{-2}\Sigma_n
    \end{bmatrix}\nonumber\\
    &=\begin{bmatrix}
        (1 + \gamma^2)\Sigma_n & 0\\
        0 & \rbr{1 + \gamma^{-2}}\Sigma_n
    \end{bmatrix}.\label{eq:thin_variance}
\end{align}
The covariance matrix exists and is invertible by \cref{cond:berry_esseen}(i), with minimum eigenvalue $\gamma_{\mathrm{min}}\eigmin{\Sigma_n}$, where     $\gamma_{\mathrm{min}} := \min\cbr{1 + \gamma^2, 1 + \gamma^{-2}}$.
Recalling the definitions of $Z^{(1)}_n$ and $Z^{(2)}_n$ in \cref{eq:define_thins_normal}, we note that similarly $Z_n^{(1,2)} := \begin{pmatrix} Z^{(1)}_n \\ Z^{(2)}_n \end{pmatrix}$ is multivariate normal with the same mean and covariance.

Applying \cref{lemma:extended_be}, under \cref{cond:berry_esseen}(ii) which ensures that the centered third moment is uniformly bounded,
\begin{align*}
    \abr{
    \Pr\rbr{\tilde Z_n^{(1, 2)} \in C} - \Pr\rbr{Z^{(1, 2)}_n \in C}
    } = O\rbr{K p^{1/4} n^{-1/2} \gamma_{\mathrm{min}}^{-3/2} \eigmin{\Sigma_n}^{-3/2}}.
\end{align*}
Since $\gamma_{\mathrm{min}} \geq 1$, it holds that $\gamma_{\mathrm{min}}^{-3/2} \leq 1$ and so we obtain our final bound
\begin{align*}
    \abr{
    \Pr\rbr{\tilde Z_n^{(1, 2)} \in C} - \Pr\rbr{Z^{(1, 2)}_n \in C}
    } = O\rbr{K p^{1/4} n^{-1/2} \eigmin{\Sigma_n}^{-3/2}}.
\end{align*}
\end{proof}

\section{Proofs of main text results}

\subsection{Proof of \cref{lemma:thinning_core_stat}}
\begin{proof}

Define
\begin{align*}
    \hat Z_n^{(1, 2)} &:=
    \begin{pmatrix}
        \hat Z_n^{(1)} \\ \hat Z_n^{(2)}
    \end{pmatrix}.
\end{align*}
\cref{lemma:thinning_core_stat_known_var} established a Berry-Esseen-type bound for thinned variables $\tilde Z_n^{(1, 2)}$ constructed using the \emph{known} variance $\Sigma_n$. We use this helper lemma to establish a bound on the thinned variables $\hat Z_n^{(1, 2)}$ constructed using an \emph{estimate} $\hat \Sigma_n$ of $\Sigma_n$. 

First, by the triangle inequality, and for $\tilde Z_n^{(1, 2)}$ defined in \cref{lemma:thinning_core_stat_known_var},
\begin{align}
    &\abr{
    \Pr\rbr{\hat Z_n^{(1, 2)} \in C} - \Pr\rbr{Z^{(1, 2)}_n \in C}
    }\nonumber\\
    &\leq
    \abr{
    \Pr\rbr{\hat Z_n^{(1, 2)} \in C} - \Pr\rbr{\tilde Z^{(1, 2)}_n \in C}
    } +
    \abr{
    \Pr\rbr{\tilde Z_n^{(1, 2)} \in C} - \Pr\rbr{Z^{(1, 2)}_n \in C}
    }\nonumber\\
    &=
    \abr{
    \Pr\rbr{\hat Z_n^{(1, 2)} \in C} - \Pr\rbr{\tilde Z^{(1, 2)}_n \in C}
    } + O\rbr{K p^{1/4} n^{-1/2} \eigmin{\Sigma_n}^{-3/2}},\label{eq:be1_bound0}
\end{align}
where the last line follows from \cref{eq:thinning_core_stat_known_var}. 

For any $\Delta > 0$, we apply the triangle inequality to the remaining term using the intermediate variable
\begin{equation*}
    \hat Z_n^{(1, 2), \Delta}
    :=
    \begin{cases} 
      \hat Z_n^{(1, 2)}, & \mathrm{if} \quad\II\sbr{\nbr{\hat \Sigma_n \Sigma_n^{-1} - I_p}_F \leq \Delta}\\
      \tilde Z_n^{(1, 2)}, & \mathrm{if}\quad \II\sbr{\nbr{\hat \Sigma_n \Sigma_n^{-1} - I_p}_F > \Delta}
   \end{cases}.
\end{equation*}
When the estimated variance $\hat \Sigma_n$ is ``close" to $\Sigma_n$, then $\hat Z_n^{(1, 2), \Delta}$ is equal to $\hat Z_n^{(1, 2)}$. Otherwise, it is defined to be $\tilde Z_n^{(1, 2)}$. The triangle inequality yields
\begin{align}
    &\abr{
    \Pr\rbr{\hat Z_n^{(1, 2)} \in C} - \Pr\rbr{\tilde Z^{(1, 2)}_n \in C}
    }\nonumber\\
    &\leq
    \abr{
    \Pr\rbr{\hat Z_n^{(1, 2)} \in C} - \Pr\rbr{\hat Z^{(1, 2), \Delta}_n \in C}
    }
    +
    \abr{
    \Pr\rbr{\hat Z_n^{(1, 2), \Delta} \in C} - \Pr\rbr{\tilde Z^{(1, 2)}_n \in C}
    }\label{eq:dtv_1}.
\end{align}
Examining the first term in \cref{eq:dtv_1}, since $\hat Z_n^{(1, 2)}$ and $\hat Z^{(1, 2), \Delta}_n$ are equal conditional on $\cbr{\nbr{\hat \Sigma_n \Sigma_n^{-1} - I_p}_F \leq \Delta}$,
\begin{align}
    \abr{
    \Pr\rbr{\hat Z_n^{(1, 2)} \in C} - \Pr\rbr{\hat Z^{(1, 2), \Delta}_n \in C}
    }
    &\leq \Pr\rbr{\nbr{\hat \Sigma_n \Sigma_n^{-1} - I_p}_F > \Delta}\nonumber\\
    &\leq \Pr\rbr{\frac{\nbr{\hat \Sigma_n - \Sigma_n}_F}{\eigmin{\Sigma_n}} > \Delta}\nonumber\\
    &\leq \Pr\rbr{\nbr{\hat \Sigma_n - \Sigma_n}_F > \Delta \eigmin{\Sigma_n}}\label{eq:be1_bound1}
\end{align}
since
\begin{equation*}
    \nbr{\hat \Sigma_n \Sigma_n^{-1} - I_p}_F
    = \nbr{\rbr{\hat \Sigma_n  - \Sigma_n} \Sigma_n^{-1}}_F
    \leq \eigmax{\Sigma_n^{-1}} \nbr{\hat \Sigma_n - \Sigma_n}_F
    = \frac{\nbr{\hat \Sigma_n - \Sigma_n}_F}{\eigmin{\Sigma_n}}.
\end{equation*}
Examining the second term in \cref{eq:dtv_1}, recall that by construction $\tilde Z_n^{(1, 2)}$ and $\hat Z^{(1, 2), \Delta}_n$ are equal conditional on $Q_\Delta := \cbr{\nbr{\hat \Sigma_n \Sigma_n^{-1} - I_p}_F > \Delta}$. Thus,
\begin{align*}
    &\abr{
    \Pr\rbr{\hat Z_n^{(1, 2), \Delta} \in C} - \Pr\rbr{\tilde Z^{(1, 2)}_n \in C}
    }\\
    &=\abr{
    \Pr\rbr{\hat Z_n^{(1, 2), \Delta} \in C, Q_\Delta} + \Pr\rbr{\hat Z_n^{(1, 2), \Delta} \in C, Q_\Delta^c} - \Pr\rbr{\tilde Z^{(1, 2)}_n \in C, Q_\Delta} - \Pr\rbr{\tilde Z_n^{(1, 2)} \in C, Q_\Delta^c}
    }\\
    &=\abr{
    \Pr\rbr{\tilde Z_n^{(1, 2)} \in C, Q_\Delta} + \Pr\rbr{\hat Z_n^{(1, 2)} \in C, Q_\Delta^c} - \Pr\rbr{\tilde Z^{(1, 2)}_n \in C, Q_\Delta} - \Pr\rbr{\tilde Z_n^{(1, 2)} \in C, Q_\Delta^c}
    }\\
    &=\abr{
    \Pr\rbr{\hat Z_n^{(1, 2)} \in C, Q_\Delta^c} - \Pr\rbr{\tilde Z_n^{(1, 2)} \in C, Q_\Delta^c}
    }\\
    &= \abr{
    \Pr\rbr{\hat Z_n^{(1, 2)} \in C \mid Q_\Delta^c} \Pr\rbr{Q_\Delta^c} - \Pr\rbr{\tilde Z^{(1, 2)}_n \in C \mid Q_\Delta^c} \Pr\rbr{Q_\Delta^c}
    }\\
    &= \Pr\rbr{Q_\Delta^c} \abr{
    \Pr\rbr{\hat Z_n^{(1, 2)} \in C \mid Q_\Delta^c} - \Pr\rbr{\tilde Z^{(1, 2)}_n \in C \mid Q_\Delta^c} 
    }\\
    &\leq \abr{
    \Pr\rbr{\hat Z_n^{(1, 2)} \in C \mid Q_\Delta^c} - \Pr\rbr{\tilde Z^{(1, 2)}_n \in C \mid Q_\Delta^c}
    }
    \stepcounter{equation}\tag{\theequation}\label{eq:dtv_2}.
\end{align*}

We now proceed to bound \cref{eq:dtv_2}. $\hat Z_n^{(1, 2)}$ and $\tilde Z^{(1, 2)}_n$ are the same up to the addition of auxiliary noise $\hat W_n \sim \Norm(0, \hat \Sigma_n)$ versus $W_n \sim \Norm(0, \Sigma_n)$. Theorem 1.1. of
\citet{devroye_total_2023} allows us to bound the total variation between two Gaussian random variables with the same mean but different covariances. To apply this result to $\hat W_n$ and $W_n$, we must condition on $\hat \Sigma_n$ so that $\hat W_n$ is Gaussian. Since $\hat \Sigma_n$ is a function of the data $D_n := \cbr{(X_i, Y_i)}_{i=1}^n$, it is sufficient to condition on $D_n$. By Theorem 1.1 and Equation 2 of \citet{devroye_total_2023}, this yields the bound
\begin{equation}\label{eq:dtv_bound}
    d_{TV}\rbr{W_n \mid D_n = d, \hat W_n \mid D_n = d} =
    O\rbr{\nbr{\hat \Sigma_n \Sigma_n^{-1} - I_p}_F},
\end{equation}
where here $\hat \Sigma_n$ depends on the data $d$; when $d$ satisfies $Q_\Delta^c$, this bound is $O(\Delta)$. 
To apply \cref{eq:dtv_bound} to \cref{eq:dtv_2}, we expand
\begin{align*}
    &\abr{
    \Pr\rbr{\hat Z_n^{(1, 2)} \in C \mid Q_\Delta^c} - \Pr\rbr{\tilde Z^{(1, 2)}_n \in C \mid Q_\Delta^c}
    }\\
    &\leq \sup_{A, B} \abr{
        \Pr\rbr{\hat Z_n \in A, \hat W_n \in B \mid Q_\Delta^c} 
        - \Pr\rbr{
        \hat Z_n \in A, W_n \in B \mid Q_\Delta^c}
    }\\
    &\leq \sup_{A, B, \rbr{D : Q_\Delta^c}} \abr{
        \Pr\rbr{ \hat W_n \in B \mid \hat Z_n \in A, D_n \in D} 
        - \Pr\rbr{
         W_n \in B \mid \hat Z_n \in A, D_n \in D}
    }\\
    &\leq \sup_{B, \rbr{D : Q_\Delta^c}} \abr{
        \Pr\rbr{\hat W_n \in B \mid D_n \in D} 
        - \Pr\rbr{
        W_n \in B \mid D_n \in D}
    }\\
    & \leq O\rbr{\Delta},
\end{align*}
where the second line uses that $\hat Z_n^{(1,2)}$ and $\tilde Z_n^{(1,2)}$ are functions of $\rbr{\hat Z_n, \hat W_n}$ and $\rbr{\hat Z_n, W_n}$, respectively; the third line holds since we can always take $D$ to be all datasets satisfying the event $Q_\Delta^c$; the fourth line uses that $\hat Z_n$ is a deterministic function of $D_n$; and the final line follows from \cref{eq:dtv_bound} for any data satisfying $Q_\Delta^c$.

Combining this with \cref{eq:be1_bound0}, \cref{eq:dtv_1}, and \cref{eq:be1_bound1}, for any $\Delta > 0$,
\begin{align*}
    \abr{
    \Pr\rbr{\hat Z_n^{(1, 2)} \in C} - \Pr\rbr{Z^{(1, 2)}_n \in C}
    }
    &\leq \Pr\rbr{\nbr{\hat \Sigma_n - \Sigma_n}_F > \Delta \eigmin{\Sigma_n}}\\
    &\quad+ O\rbr{K p^{1/4} n^{-1/2}\eigmin{\Sigma_n}^{-3/2} + \Delta}.
\end{align*}


\end{proof}

\subsection{Proof of \cref{lemma:conditional_berry_esseen}}
\begin{proof}
Our goal is to bound
\begin{equation}\label{eq:conditional_diff}
    \abr{
    \Pr \rbr{ \hat Z_n^{(2)} + r_n^B \in B_n \mid\, \hat Z_n^{(1)} + r_n^A \in A_n}
    -
    \Pr(Z_n^{(2)} \in B_n)
    }.
\end{equation}

For ease of notation, recall $\hat Z_n^{(1, 2)} := \begin{pmatrix} \hat Z^{(1)}_n \\ \hat Z^{(2)}_n \end{pmatrix}$, and let $C_n := A_n \times B_n$ and $r_n := \rbr{r_n^A, r_n^B}$. To bound the conditional probability in \cref{eq:conditional_diff}, we study the following bound on the joint probability:
\begin{align}
    &\abr{\Pr \rbr{\hat Z^{(1)}_n + r_n^A\in A_n, \hat Z^{(2)}_n + r_n^B\in B_n} - \Pr \rbr{Z^{(1)}_n \in A_n, Z^{(2)}_n \in B_n}}\label{eq:joint_diff}\\
    &\leq \abr{\Pr\rbr{
        \hat Z_n^{(1, 2)} + r_n \in C_n
    }
    - 
    \Pr\rbr{
        \hat Z_n^{(1, 2)} \in C_n
    }}\label{eq:rn_bound1}\\
    &\quad+ \abr{\Pr \rbr{\hat Z^{(1, 2)}_n \in C_n} - \Pr \rbr{Z^{(1, 2)}_n \in C_n}}.\label{eq:rn_bound2}
\end{align}

Since by statement of the Lemma, $C_n$ can be written as the union of $K$ almost surely disjoint, measurable convex sets, we can bound \cref{eq:rn_bound2} using the Berry-Esseen-type result of \cref{lemma:thinning_core_stat}, yielding
\begin{align}
    \abr{\Pr \rbr{\hat Z^{(1, 2)}_n \in C_n} - \Pr \rbr{Z^{(1, 2)}_n \in C_n}}
    &\leq
    \Pr\rbr{\nbr{\hat \Sigma_n - \Sigma_n}_F > \Delta \eigmin{\Sigma_n}}\nonumber\\
    &\quad+ O\rbr{K p^{1/4} n^{-1/2}\eigmin{\Sigma_n}^{-3/2} + \Delta}\label{eq:b7_bound}
\end{align}
for any $\Delta > 0$.

The key technical challenge remaining is to bound \cref{eq:rn_bound1}. Since by assumption we can write $C_n$ as the union of the $K$ almost surely disjoint convex sets $\cbr{C_{n, k}}_{k=1}^K$, we can first upper-bound \cref{eq:rn_bound1} using the triangle inequality:
\begin{align}
    &\abr{\Pr\rbr{
        \hat Z_n^{(1, 2)} + r_n \in C_n
    }
    - 
    \Pr\rbr{
        \hat Z_n^{(1, 2)} \in C_n
    }}\nonumber\\
    &= \abr{\sum_{k=1}^K \Pr\rbr{
        \hat Z_n^{(1, 2)} + r_n \in C_{n, k}
    }
    - 
    \sum_{k=1}^K \Pr\rbr{
        \hat Z_n^{(1, 2)} \in C_{n, k}
    }}\nonumber\\
    &\leq \sum_{k=1}^K \abr{\Pr\rbr{
        \hat Z_n^{(1, 2)} + r_n \in C_{n, k}
    }
    - 
    \Pr\rbr{
        \hat Z_n^{(1, 2)} \in C_{n, k}
    }}.
    \label{eq:rn_bound1_union}
\end{align}

To bound each of the $K$ terms in \cref{eq:rn_bound1_union}, we will construct upper bounds for each of the two directions, positive or negative, that the term inside the absolute value could take.

\paragraph{Upper bounding the positive direction:} 
We first construct an upper bound on
\begin{align*}
    \Pr\rbr{
     \hat Z_n^{(1, 2)} + r_n \in C_{n, k}
    }
    - 
    \Pr\rbr{
    \hat Z_n^{(1, 2)} \in C_{n, k}
    }.
\end{align*}

For any $\varepsilon > 0$, let $\BB_\varepsilon$ denote an $\varepsilon$-ball in $\RR^{2p}$ and let 
\begin{equation}\label{eq:eps_ball}
    C_{n, k}^\varepsilon := C_{n, k} + \BB_\varepsilon = \cbr{c + r : c \in C_{n, k}, \nbr{r} \leq \varepsilon};
\end{equation}
$C_{n, k}^\varepsilon$ remains a convex set.
We obtain the bound
\begin{align}
    &\Pr\rbr{
    \hat Z_n^{(1, 2)} + r_n \in C_{n, k}
    }
    - 
    \Pr\rbr{
    \hat Z_n^{(1, 2)} \in C_{n, k}
    }\nonumber\\
    &\leq \Pr\rbr{
    \hat Z_n^{(1, 2)} + r_n \in C_{n, k}, r_n \in \BB_\varepsilon
    }
    - 
    \Pr\rbr{
    \hat Z_n^{(1, 2)} \in C_{n, k}
    }
    +
    \Pr\rbr{r_n \not\in \BB_\varepsilon}\nonumber\\
    &\leq \Pr\rbr{
    \hat Z_n^{(1, 2)} \in C_{n, k}^\varepsilon, r_n \in \BB_\varepsilon
    }
    - 
    \Pr\rbr{
    \hat Z_n^{(1, 2)} \in C_{n, k}^\varepsilon
    }
    +
    \Pr\rbr{r_n \not\in \BB_\varepsilon}\nonumber\\
    &\leq \Pr\rbr{
    \hat Z_n^{(1, 2)} \in C_{n, k}^\varepsilon
    }
    - 
    \Pr\rbr{
    \hat Z_n^{(1, 2)} \in C_{n, k}
    }
    +
    \Pr\rbr{r_n \not\in \BB_\varepsilon}\nonumber\\
    &= \Pr\rbr{
    \hat Z_n^{(1, 2)} \in C_{n, k}^\varepsilon
    }
    - 
    \Pr\rbr{
    Z_n^{(1, 2)} \in C_{n, k}^\varepsilon
    }
    +
    \Pr\rbr{
    \hat Z_n^{(1, 2)} \in C_{n, k}
    }
    - 
    \Pr\rbr{
    Z_n^{(1, 2)} \in C_{n, k}
    }\nonumber\\
    &\quad +
    \Pr\rbr{
    Z_n^{(1, 2)} \in C_{n, k}^\varepsilon
    }
    -
    \Pr\rbr{
    Z_n^{(1, 2)} \in C_{n, k}
    }
    +
    \Pr\rbr{r_n \not\in \BB_\varepsilon}\nonumber\\
    &\leq \abr{\Pr\rbr{
    \hat Z_n^{(1, 2)} \in C_{n, k}^\varepsilon
    }
    - 
    \Pr\rbr{
    Z_n^{(1, 2)} \in C_{n, k}^\varepsilon
    }}
    +
    \abr{\Pr\rbr{
    \hat Z_n^{(1, 2)} \in C_{n, k}
    }
    - 
    \Pr\rbr{
    Z_n^{(1, 2)} \in C_{n, k}
    }}\label{eq:case1_bound_diff}\\
    &\quad +
    \Pr\rbr{
    Z_n^{(1, 2)} \in C_{n, k}^\varepsilon \setminus C_{n, k}
    }
    +
    \Pr\rbr{r_n \not\in \BB_\varepsilon}\label{eq:case1_bound_perim}.
\end{align}
Since $C_{n, k}$ and $C_{n, k}^\varepsilon$ are each convex sets, for any $\Delta > 0$ we can bound \cref{eq:case1_bound_diff} using \cref{lemma:thinning_core_stat}:
\begin{align*}
    &\abr{\Pr\rbr{
    \hat Z_n^{(1, 2)} \in C_{n, k}^\varepsilon
    }
    - 
    \Pr\rbr{
    Z_n^{(1, 2)} \in C_{n, k}^\varepsilon
    }}
    + 
    \abr{\Pr\rbr{
    \hat Z_n^{(1, 2)} \in C_{n, k}
    }
    - 
    \Pr\rbr{
    Z_n^{(1, 2)} \in C_{n, k}
    }}\\
    &= 2\Pr\rbr{\nbr{\bar \Sigma_n - \Sigma_n}_F > \Delta \eigmin{\Sigma_n}} + O\rbr{p^{1/4} n^{-1/2} + \Delta}.
\end{align*}
Bounding the first term in \cref{eq:case1_bound_perim} is the problem of bounding the $\varepsilon$-Gaussian perimeter.
We use the result of \citet{nazarov_maximal_2003}, formally stated by \citet{chen_multivariate_2015} as follows. First, let $C_{n, k}^{-\varepsilon}$ denote the $\varepsilon$-shrinkage of a set: 
\begin{equation}\label{eq:epsilon_shrinkage}
    C_{n, k}^{-\varepsilon} := \cbr{c \in C_{n, k} : c + \BB_\varepsilon \subset C_{n, k}};
\end{equation}
this remains a convex set.
\begin{proposition}[Adapted from Proposition 2.5 \citep{chen_multivariate_2015}]\label{prop:chen_2.5}
    For any convex set $C_{n, k} \in \RR^{2p}$ and $\varepsilon_1, \varepsilon_2 \geq 0$,
    \begin{equation*}
        \Pr\rbr{\rbr{\Sigma_n^{(1,2)}}^{-1/2}\rbr{Z_n^{(1, 2)} - \mu_n^{(1,2)}}\in C_{n, k}^{\varepsilon_1} \setminus C_{n, k}^{-\varepsilon_2}} \leq (2p)^{1/2} (\varepsilon_1 + \varepsilon_2).
    \end{equation*}
\end{proposition}
We now wish to apply \cref{prop:chen_2.5} to bound $\Pr\rbr{
    Z_n^{(1, 2)} \in C_n^{\varepsilon_1} \setminus C_n^{-\varepsilon_2}
    }$ for any $\varepsilon_1, \varepsilon_2 \geq 0$,
and thus must standardize $Z_n^{(1,2)}$. Let $\tilde C_{n, k} := \cbr{\rbr{\Sigma_n^{(1,2)}}^{-1/2}\rbr{c - \mu_n^{(1,2)}}: c \in C_{n, k}}$ and observe that
\begin{align*}
    &\cbr{\rbr{\Sigma_n^{(1,2)}}^{-1/2}\rbr{c - \mu_n^{(1,2)}}: c \in C_{n, k}^{\varepsilon_1}}\\
    &= \cbr{\rbr{\Sigma_n^{(1,2)}}^{-1/2}\rbr{c + r - \mu_n^{(1,2)}}: c \in C_{n, k}, \nbr{r}_2 \leq \varepsilon_1}\\
    &= \cbr{\rbr{\Sigma_n^{(1,2)}}^{-1/2}\rbr{c - \mu_n^{(1,2)}} + \rbr{\Sigma_n^{(1,2)}}^{-1/2}r: c \in C_{n, k}, \nbr{r}_2 \leq \varepsilon_1}\\
    &\subseteq \cbr{\rbr{\Sigma_n^{(1,2)}}^{-1/2}\rbr{c - \mu_n^{(1,2)}} + r: c \in C_{n, k}, \nbr{r}_2 \leq \varepsilon_1 \eigmax{\rbr{\Sigma_n^{(1,2)}}^{-1/2}}}\\
    &= \tilde C_{n, k}^{\varepsilon_1\eigmax{\rbr{\Sigma_n^{(1,2)}}^{-1/2}}}\\
    &= \tilde C_{n, k}^{\varepsilon_1\sqrt{\eigmin{\Sigma_n^{(1,2)}}}}\\
    &\subseteq \tilde C_{n, k}^{\varepsilon_1\sqrt{2\eigmin{\Sigma_n}}},
\end{align*}
where the second line follows from the definition of $C_{n, k}^{\varepsilon_1}$ in \cref{eq:eps_ball}, and the last line follows from the form of the variance $\Sigma_n^{(1,2)}$, given in \cref{eq:thin_variance}.

Similarly,
\begin{align*}
    &\cbr{\rbr{\Sigma_n^{(1,2)}}^{-1/2}\rbr{c - \mu_n^{(1,2)}}: c \in C_{n, k}^{-\varepsilon_2}}\\
    &= \cbr{\rbr{\Sigma_n^{(1,2)}}^{-1/2}\rbr{c - \mu_n^{(1,2)}}: c + \BB_{\varepsilon_2} \subset C_{n, k}}\\
    &= \cbr{c: c + \cbr{\rbr{\Sigma_n^{(1,2)}}^{-1/2}r : r \in \BB_{\varepsilon_2}}\subset \tilde C_{n, k}}\\
    &\supseteq \cbr{c: c + \cbr{\eigmax{\rbr{\Sigma_n^{(1,2)}}^{-1/2}}r : r \in \BB_{\varepsilon_2}}\subset \tilde C_{n, k}}\\
    &= \cbr{c: c + \cbr{\eigmax{\rbr{\Sigma_n^{(1,2)}}^{-1/2}}r : r \in \BB_{\varepsilon_2}}\subset \tilde C_{n, k}}\\
    &= \tilde C_{n, k}^{-\varepsilon_2 / \eigmax{\rbr{\Sigma_n^{(1,2)}}^{-1/2}}}\\
    &= \tilde C_{n, k}^{-\varepsilon_2 / \sqrt{\eigmin{\Sigma_n^{(1,2)}}}} \\
    &\supseteq \tilde C_{n, k}^{-\varepsilon_2 / \sqrt{2\eigmin{\Sigma_n}}}.
\end{align*}

Combining this with \cref{prop:chen_2.5}, we see that for any $\varepsilon_1, \varepsilon_2 \geq 0$,
\begin{align}
    &\Pr\rbr{
    Z_n^{(1, 2)} \in C_{n, k}^{\varepsilon_1} \setminus C_{n, k}^{-\varepsilon_2}
    }\nonumber\\
    &\leq \Pr\rbr{
    \rbr{\Sigma_n^{(1,2)}}^{-1/2}\rbr{Z_n^{(1, 2)} - \mu_n^{(1,2)}} \in \tilde C_{n, k}^{\varepsilon_1\sqrt{2\eigmin{\Sigma_n}}} \setminus \tilde C_{n, k}^{-\varepsilon_2 / \sqrt{2\eigmin{\Sigma_n}}}
    }\nonumber\\
    &\leq (2p)^{1/2}\rbr{\varepsilon_1\sqrt{2\eigmin{\Sigma_n}} + \varepsilon_2 / \sqrt{2\eigmin{\Sigma_n}}}\nonumber\\
    &= 2p^{1/2}\rbr{\varepsilon_1 \eigmin{\Sigma_n}^{1/2} + \varepsilon_2 \eigmin{\Sigma_n}^{-1/2}}.\label{eq:prop2.5_expanded}
\end{align}
This bounds the first term in \cref{eq:case1_bound_perim} as follows:
\begin{equation*}
    \Pr\rbr{
    Z_n^{(1, 2)} \in C_{n, k}^\varepsilon \setminus C_{n, k}
    } \leq 2p^{1/2}\varepsilon \eigmin{\Sigma_n}^{1/2}.
\end{equation*}

Putting everything together, for any $\varepsilon, \Delta > 0$, we obtain the bound
\begin{align}
    &\Pr\rbr{
        \hat Z_n^{(1, 2)} + r_n \in C_{n, k}
    }
    - 
    \Pr\rbr{
        \hat Z_n^{(1, 2)} \in C_{n, k}
    }\nonumber\\
    &\leq O\rbr{\Pr\rbr{\nbr{\hat \Sigma_n - \Sigma_n}_F > \Delta \eigmin{\Sigma_n}} + p^{1/4} n^{-1/2}\eigmin{\Sigma_n}^{-3/2} + \Delta + 2p^{1/2}\varepsilon \eigmin{\Sigma_n}^{1/2}}\nonumber\\
    &\quad+ \Pr\rbr{r_n \not\in \BB_\varepsilon}\label{eq:abs_bound1},
\end{align}
absorbing constants into the big-$O$.

\paragraph{Upper bounding the negative direction:} We now construct an upper bound in the other sign direction on
\begin{align*}
    \Pr\rbr{
    \hat Z_n^{(1, 2)} \in C_{n, k}
    }
    -
    \Pr\rbr{
    \hat Z_n^{(1, 2)} + r_n \in C_{n, k}
    }.
\end{align*}
Let 
\begin{equation*}
    \mathrm{inrad}(C_{n, k}) := \max\cbr{\varepsilon' : \exists c \in C_{n, k} \quad\mathrm{s.t.}\, c + \BB_{\varepsilon'} \subset C_{n, k}}
\end{equation*}
denote the inradius of its argument set $C_{n, k}$, and let
\begin{equation}\label{eq:inradius}
    \varepsilon \leq \min \cbr{\cup_{k=1}^K\cbr{\mathrm{inrad}(C_{n, k})}}.   
\end{equation}
We now use the fact that for any $r \in \BB_\varepsilon$ and $z \in \RR^{2p}$,
\begin{equation}\label{eq:indicator_identity}
    \II\sbr{z \in C_{n, k}} - \II\sbr{z + r_n \in C_{n, k}} \leq \II\sbr{z \in C_{n, k} \setminus C_{n, k}^{-\nbr{r_n}_2}},
\end{equation}
where $C_{n, k}^{-\nbr{r}_2}$ was defined in \cref{eq:epsilon_shrinkage}. \cref{eq:indicator_identity} holds because when the right-hand side indicator of \cref{eq:indicator_identity} is $0$, either (i) $z \not\in C_{n, k}$ and thus $ \II\sbr{z \in C_{n, k}} = 0$ on the left-hand side, or (ii) $z \in C_{n, k}^{-\nbr{r}_2}$, which is nonempty since $r \leq \varepsilon$, and so $\II\sbr{z + r \in A_n} = 1$ on the left-hand side. This yields the bound
\begin{align}
    &\Pr\rbr{
    \hat Z_n^{(1, 2)} \in C_{n, k}
    }
    -
    \Pr\rbr{
    \hat Z_n^{(1, 2)} + r_n \in C_{n, k}
    }\nonumber\\
    &\leq\Pr\rbr{
    \hat Z_n^{(1, 2)} \in C_{n, k}, r_n \in \BB_\varepsilon
    }
    -
    \Pr\rbr{
    \hat Z_n^{(1, 2)} + r_n \in C_{n, k}, r_n \in \BB_\varepsilon
    }
    +
    \Pr\rbr{r_n \not\in \BB_\varepsilon}\label{eq:use_ltp}\\
    &\leq \Pr\rbr{\hat Z_n^{(1, 2)} \in C_{n, k} \setminus C_{n, k}^{-\nbr{r_n}_2}, r_n \in \BB_\varepsilon} + \Pr\rbr{r_n \not\in \BB_\varepsilon}\label{eq:use_indic}\\
    &\leq \Pr\rbr{\hat Z_n^{(1, 2)} \in C_{n, k} \setminus C_{n, k}^{-\varepsilon}, r_n \in \BB_\varepsilon} + \Pr\rbr{r_n \not\in \BB_\varepsilon}\label{eq:use_ub}\\
    &\leq \Pr\rbr{\hat Z_n^{(1, 2)} \in C_{n, k} \setminus C_{n, k}^{-\varepsilon}} + \Pr\rbr{r_n \not\in \BB_\varepsilon}\nonumber\\
    &= \Pr\rbr{\hat Z_n^{(1, 2)} \in C_{n, k}} - \Pr\rbr{\hat Z_n^{(1, 2)} \in C_{n, k}^{-\varepsilon}} + \Pr\rbr{r_n \not\in \BB_\varepsilon}\nonumber\\
    &= \Pr\rbr{\hat Z_n^{(1, 2)} \in C_{n, k}} - \Pr\rbr{Z_n^{(1, 2)} \in C_{n, k}} + \Pr\rbr{Z_n^{(1, 2)} \in C_{n, k}^{-\varepsilon}} - \Pr\rbr{\hat Z_n^{(1, 2)} \in C_{n, k}^{-\varepsilon}}\nonumber\\
    &\quad+ \Pr\rbr{Z_n^{(1, 2)} \in C_{n, k}} - \Pr\rbr{Z_n^{(1, 2)} \in C_{n, k}^{-\varepsilon}} + \Pr\rbr{r_n \not\in \BB_\varepsilon}\label{eq:use_as1}\\
    &\leq \abr{\Pr\rbr{\hat Z_n^{(1, 2)} \in C_{n, k}} - \Pr\rbr{Z_n^{(1, 2)} \in C_{n, k}}} + \abr{\Pr\rbr{\hat Z_n^{(1, 2)} \in C_{n, k}^{-\varepsilon}} - \Pr\rbr{Z_n^{(1, 2)} \in C_{n, k}^{-\varepsilon}}}\nonumber\\
    &\quad+ \Pr\rbr{Z_n^{(1, 2)} \in C_{n, k} \setminus C_{n, k}^{-\varepsilon}}+ \Pr\rbr{r_n \not\in \BB_\varepsilon}\nonumber\\
    &\leq O\rbr{\Pr\rbr{\nbr{\bar \Sigma_n - \Sigma_n}_F > \Delta \eigmin{\Sigma_n}} + p^{1/4} n^{-1/2}\eigmin{\Sigma_n}^{-3/2} + \Delta + 2p^{1/2}\varepsilon \eigmin{\Sigma_n}^{-1/2}}\nonumber\\
    &\quad + \Pr\rbr{r_n \not\in \BB_\varepsilon},\label{eq:use_lemma1}
\end{align}
where \cref{eq:use_ltp} applies the law of total probability, \cref{eq:use_indic} applies the indicator bound in \cref{eq:indicator_identity}, \cref{eq:use_ub} uses the fact that $C_{n, k}^{-\varepsilon} \subseteq C_{n, k}^{-\nbr{r_n}_2}$ when $\nbr{r_n}_2 \leq \varepsilon$, \cref{eq:use_as1} simply adds terms summing to zero, and \cref{eq:use_lemma1} applies \cref{lemma:thinning_core_stat} and \cref{eq:prop2.5_expanded}.

\paragraph{Bounding \cref{eq:joint_diff}:}  

Having upper bounded the two possible directions that the term inside of \cref{eq:rn_bound1_union} could go, we obtain a bound on the sum of $K$ absolute differences in \cref{eq:rn_bound1}:
\begin{align*}
    &\abr{\Pr \rbr{\hat Z^{(1)}_n + r_n^A \in A_n, \hat Z^{(2)}_n + r_n^B \in B_n} - \Pr \rbr{\hat Z^{(1)}_n \in A_n, \hat Z^{(2)}_n \in B_n}}\\
    &\leq K O\rbr{\Pr\rbr{\nbr{\hat \Sigma_n - \Sigma_n}_F > \Delta \eigmin{\Sigma_n}} + p^{1/4} n^{-1/2}\eigmin{\Sigma_n}^{-3/2} + \Delta + p^{1/2}\varepsilon\rbr{\eigmin{\Sigma_n}^{1/2} + \eigmin{\Sigma_n}^{-1/2}}}\\
    &\quad+ K \Pr\rbr{r_n \not\in \BB_\varepsilon}.
\end{align*}

We go further and note that
\begin{align*}
    \Pr\rbr{r_n \not\in \BB_\varepsilon}
    &= \Pr\rbr{\nbr{r_n} > \varepsilon}\\
    &\leq \Pr\rbr{\nbr{r_n^A}_2 > \varepsilon/\sqrt{2} \cup \nbr{r_n^B}_2 > \varepsilon/\sqrt{2}}\\
    &\leq \Pr\rbr{\nbr{r_n^A}_2 > \varepsilon/\sqrt{2}} + \Pr\rbr{\nbr{r_n^B}_2 > \varepsilon/\sqrt{2}}\\
    &\leq \Pr\rbr{\nbr{r_n^A}_2 > \tfrac\varepsilon2} + \Pr\rbr{\nbr{r_n^B}_2 > \tfrac\varepsilon2}.
\end{align*}

Combining the above bound on \cref{eq:rn_bound1} with the earlier derived bound \cref{eq:b7_bound} for \cref{eq:rn_bound2}, we obtain a final bound on the joint probability \cref{eq:joint_diff}:
\begin{align}
    &\abr{\Pr \rbr{\hat Z^{(1)}_n + r_n^A\in A_n, \hat Z^{(2)}_n + r_n^B\in B_n} - \Pr \rbr{Z^{(1)}_n \in A_n, Z^{(2)}_n \in B_n}}\label{eq:boundo_on_b7}\\
    &\leq K O\rbr{\Pr\rbr{\nbr{\hat \Sigma_n - \Sigma_n}_F > \Delta \eigmin{\Sigma_n}} + p^{1/4} n^{-1/2}\eigmin{\Sigma_n}^{-3/2} + \Delta + p^{1/2}\varepsilon\rbr{\eigmin{\Sigma_n}^{1/2} + \eigmin{\Sigma_n}^{-1/2}}}\nonumber\\
    &\quad+ K\Pr\rbr{2r_n^A \not\in \BB_\varepsilon} + K\Pr\rbr{2r_n^B \not\in \BB_\varepsilon}\nonumber.
\end{align}

\paragraph{Expanding the conditional probability \cref{eq:conditional_diff}:}

We now return to bounding \cref{eq:conditional_diff}:
\begin{align}
    &\abr{
        \Pr \rbr{\hat Z^{(2)}_n + r_n^B \in B_n | \hat Z^{(1)}_n + r_n^A \in A_n}
        - \Pr \rbr{Z^{(2)}\in B_n}
    }\nonumber\\
    &= 
    \abr{
        \frac{
            \Pr \rbr{\hat Z^{(2)}_n + r_n^B\in B_n, \hat Z^{(1)}_n+ r_n^A \in A_n}
            }{
            \Pr \rbr{\hat Z^{(1)}_n + r_n^A\in A_n}
            } -
            \frac{
                \Pr \rbr{Z^{(2)}_n \in B_n} \Pr \rbr{\hat Z^{(1)}_n + r_n^A\in A_n}
            }{
                \Pr \rbr{\hat Z^{(1)}_n + r_n^A\in A_n}
            } 
    }\label{eq:expand_conditional_prob}\\
    &\leq
    \abr{
        \frac{
            \Pr \rbr{\hat Z^{(2)}_n + r_n^B \in B_n, \hat Z^{(1)}_n + r_n^A \in A_n}
            }{
            \Pr \rbr{\hat Z^{(1)}_n + r_n^A \in A_n}
            } -
            \frac{
                \Pr \rbr{Z^{(2)}_n \in B_n, Z^{(1)}_n \in A_n}
            }{
                \Pr \rbr{\hat Z^{(1)}_n + r_n^A\in A_n}
            } 
    }\nonumber\\
    &\quad+
    \abr{
        \frac{
            \Pr \rbr{Z^{(2)}_n \in B_n, Z^{(1)}_n \in A_n}
        }{
            \Pr \rbr{\hat Z^{(1)}_n + r_n^A \in A_n}
        }  -
        \frac{
            \Pr \rbr{Z^{(2)}_n \in B_n} \Pr \rbr{\hat Z^{(1)}_n + r_n^A \in A_n}
        }{
            \Pr \rbr{\hat Z^{(1)}_n + r_n^A \in A_n}
        } 
    }\label{eq:apply_teq}\\
    &= 
    \frac{
        \abr{\Pr \rbr{\hat Z^{(2)}_n + r_n^B \in B_n, \hat Z^{(1)}_n + r_n^A \in A_n}
        -
        \Pr \rbr{Z^{(2)}_n \in B_n, Z^{(1)}_n \in A_n}
        }}{
        \Pr \rbr{\hat Z^{(1)}_n + r_n^A \in A_n}
        }\nonumber\\
    &\quad+
        \frac{\abr{
            \Pr \rbr{Z^{(2)}_n \in B_n} \Pr\rbr{ Z^{(1)}_n \in A_n}
            -
            \Pr \rbr{Z^{(2)}_n \in B_n}
            \Pr \rbr{\hat Z^{(1)}_n + r_n^A \in A_n}
        }}{
            \Pr \rbr{\hat Z^{(1)}_n + r_n^A \in A_n}
        }\label{eq:combine_shared_denom}\\
    &\leq
    \frac{
        \abr{\Pr \rbr{\hat Z^{(2)}_n + r_n^B \in B_n, \hat Z^{(1)}_n + r_n^A \in A_n}
        -
        \Pr \rbr{Z^{(2)}_n \in B_n, Z^{(1)}_n \in A_n}
        }}{
        \Pr \rbr{\hat Z^{(1)}_n + r_n^A \in A_n}
        }\nonumber\\
    &\quad+
        \frac{\abr{
            \Pr \rbr{Z^{(1)}_n \in A_n}
            -
            \Pr \rbr{\hat Z^{(1)}_n + r_n^A \in A_n}
        }}{
            \Pr \rbr{\hat Z^{(1)}_n + r_n^A \in A_n}
        }\label{eq:factor_out}\\
    &\leq K\frac{
    O\rbr{\Pr\rbr{\nbr{\hat \Sigma_n - \Sigma_n}_F > \Delta \eigmin{\Sigma_n}} + p^{1/4} n^{-1/2}\eigmin{\Sigma_n}^{-3/2} + \Delta + p^{1/2}\varepsilon\rbr{\eigmin{\Sigma_n}^{1/2} + \eigmin{\Sigma_n}^{-1/2}}}}{\Pr \rbr{\hat Z^{(1)}_n + r_n^A \in A_n}}\nonumber\\
    &\quad+2K\frac{\Pr\rbr{\nbr{r_n^A}_2 > \varepsilon/2} + \Pr\rbr{\nbr{r_n^B}_2 > \varepsilon/2}}
    {\Pr \rbr{\hat Z^{(1)}_n + r_n^A \in A_n}},\label{eq:apply_bounds}
\end{align}
where \cref{eq:expand_conditional_prob} expands the definition of conditional probability, \cref{eq:apply_teq} applies the triangle inequality, \cref{eq:combine_shared_denom} uses the fact that $Z_n^{(1)}$ and $Z_n^{(2)}$ are independent, \cref{eq:factor_out} factors out $\Pr \rbr{Z^{(2)}_n \in B_n} \leq 1$, and \cref{eq:apply_bounds} bounds each of the two numerator using \cref{eq:boundo_on_b7}, in the latter numerator implicitly taking $B_n = \RR^p$.

\paragraph{Bounding the denominator:}
Lastly, we lower bound the denominator of \cref{eq:apply_bounds}, $\Pr\rbr{\hat Z^{(1)}_n + r_n^A \in A_n}$:
\begin{align}
    &\Pr\rbr{\hat Z^{(1)}_n + r_n^A \in A_n}\nonumber\\
    &\geq \Pr\rbr{\hat Z^{(1)}_n + r_n^A \in A_n, \nbr{r_n^A}_2 \leq \varepsilon}\nonumber\\
    &\geq \Pr\rbr{\hat Z^{(1)}_n \in A_n^{-\varepsilon}, \nbr{r_n^A}_2 \leq \varepsilon}\nonumber\\
    &= 1 - \Pr\rbr{\cbr{\hat Z^{(1)}_n \not\in A_n^{-\varepsilon}}\, \cup\, \cbr{\nbr{r_n^A}_2 > \varepsilon}}\nonumber\\
    &\geq 1 - \abr{\Pr\rbr{\hat Z^{(1)}_n \not\in A_n^{-\varepsilon}} - \Pr\rbr{Z^{(1)}_n \not\in A_n^{-\varepsilon}}} - \Pr\rbr{Z^{(1)}_n \not\in A_n^{-\varepsilon}} - \Pr\rbr{\nbr{r_n^A}_2 > \varepsilon}\nonumber\\
    &= \Pr\rbr{Z^{(1)}_n \in A_n^{-\varepsilon}} - \abr{\Pr\rbr{\hat Z^{(1)}_n \not\in A_n^{-\varepsilon}} - \Pr\rbr{Z^{(1)}_n \not\in A_n^{-\varepsilon}}} - \Pr\rbr{\nbr{r_n^A}_2 > \varepsilon}\nonumber\\
    &\geq \Pr\rbr{Z^{(1)}_n \in A_n^{-\varepsilon}} - \Pr\rbr{\nbr{\hat \Sigma_n - \Sigma_n}_F > \Delta \eigmin{\Sigma_n}} \nonumber\\
    &\quad- O\rbr{K p^{1/4} n^{-1/2}  \eigmin{\Sigma_n}^{-3/2} + \Delta} - \Pr\rbr{\nbr{r_n^A}_2 > \varepsilon},\label{eq:denom_lb}
\end{align}
where the last line uses the bound from \cref{lemma:thinning_core_stat} on the absolute difference.

\paragraph{A final bound on \cref{eq:conditional_diff}:}

Recalling the notation in the statement of \cref{lemma:conditional_berry_esseen}, we defined the rates
\begin{align*}
    E_1 &:= p^{1/4} n^{-1/2}\eigmin{\Sigma_n}^{-3/2} + \Delta + \Pr\rbr{\nbr{\hat \Sigma_n - \Sigma_n}_F > \Delta \eigmin{\Sigma_n}},\\
    E_2 &:= O\rbr{E_1 + p^{1/2}\varepsilon\rbr{\eigmin{\Sigma_n}^{1/2} + \eigmin{\Sigma_n}^{-1/2}}} + \Pr\rbr{\nbr{r_n^A}_2 > \tfrac\varepsilon2} + \Pr\rbr{\nbr{r_n^B}_2 > \tfrac\varepsilon2},\\
    E_3 &:= \Pr\rbr{Z^{(1)}_n \in A_n^{-\varepsilon}} - O\rbr{E_1} - \Pr\rbr{\nbr{r_n^A}_2 > \varepsilon}.
\end{align*}
The bound in \cref{eq:apply_bounds} is
\begin{equation*}
    \abr{
        \Pr \rbr{\hat Z^{(2)}_n + r^B_n \in B_n \middle| \hat Z^{(1)}_n + r^A_n \in A_n}
        - \Pr \rbr{Z^{(2)}_n \in B_n}
        }
        \leq
        \frac{
        K E_2}
        {\Pr\rbr{\hat Z^{(1)}_n + r_n^A \in A_n}}.
\end{equation*}
The bound in \cref{eq:denom_lb} implies
\begin{equation*}
    \Pr\rbr{\hat Z^{(1)}_n + r_n^A \in A_n} \geq E_3.
\end{equation*}
We must take care in the case that the denominator lower bound $E_3$ becomes negative, a case which implies $E_3 < K E_2$. If $E_3 < K E_2$, we can trivially upper bound \cref{eq:conditional_diff} by one. Thus, we arrive at our final bound
\begin{equation*}
    \abr{
        \Pr \rbr{\hat Z^{(2)}_n + r^B_n \in B_n \middle| \hat Z^{(1)}_n + r^A_n \in A_n}
        - \Pr \rbr{Z^{(2)}_n \in B_n}
        }
        \leq
        \frac{
        K E_2}
        {\max\cbr{K E_2, E_3}}.
\end{equation*}

\end{proof}

\subsection{Proof of \cref{lemma:selection_event}}
\begin{proof}
By first-order optimality conditions, the penalized estimator $\MestP$, defined in \cref{eq:MestP}, solves
\begin{equation}\label{eq:selection_zero}
    \sqrt{n} P_n \dot\ell_{\MestP} + \gamma \hat W_n + \eta = 0
\end{equation}
for some subderivative $\eta \in \partial \rho_\lambda\rbr{\MestP}$, where $\partial \rho_\lambda\rbr{\MestP}$ is the sub-differential of $\rho_\lambda(\theta)$ evaluated at $\theta = \MestP$.
Scaling \cref{cond:taylor_expansion} by $\sqrt{n}$ yields
\begin{equation}\label{eq:taylor}
    \sqrt{n} P_n \dot\ell_{\MestP} = \sqrt{n} P_n \dot \ell_{\Mstar} +  \sqrt{n} H_n (\MestP - \Mstar) + r_n^A\rbr{\MestP},
\end{equation}

where $r_n^A = o_p\rbr{1}$. Combining \cref{eq:taylor} with the optimality condition in \cref{eq:selection_zero}, there exists some $\eta \in \partial \rho_\lambda\rbr{\MestP}$ such that
\begin{equation*}
    \sqrt{n} P_n \dot \ell_{\Mstar} + \gamma \hat W_n + \sqrt{n} H_n (\MestP - \Mstar) + r_n^A\rbr{\MestP} + \eta = 0,
\end{equation*}
or equivalently
\begin{equation}\label{eq:penalized_taylor}
    \sqrt{n} P_n \dot \ell_{\Mstar} + \gamma \hat W_n - H_n \sqrt{n} \Mstar + r_n^A\rbr{\MestP} = - H_n \sqrt{n} \MestP - \eta.
\end{equation}

We now use \cref{eq:penalized_taylor} to characterize our selection event.
For any $E$, let $\RR_E := \cbr{\theta \in \RR^p : \supp(\theta) = E}$.

\begin{align}
    &\cbr{
        \supp\rbr{\MestP} = E
        }\nonumber\\
    &= \cbr{
        \MestP \in \RR^p_{E}
        }\nonumber\\
    &= \cbr{
        \sqrt{n}\MestP \in \RR^p_{E}
        }\nonumber\\
    &= \cbr{
        -\sqrt{n}H_n\MestP\in \cbr{-H_n \theta  : \theta \in \RR^p_{E}}
        }\nonumber\\
    &= \cbr{
        -\sqrt{n}H_n\MestP - \eta \in \cbr{-H_n \theta - \tilde \eta : \theta \in \RR^p_{E}, \tilde \eta \in \partial \rho_\lambda(\theta)}
        }\nonumber\\
    &= \cbr{
        \sqrt{n} P_n \dot \ell_{\Mstar} + \gamma \hat W_n - H_n \sqrt{n} \Mstar + r_n^A\rbr{\MestP} \in \cbr{-H_n \theta - \tilde \eta : \theta \in \RR^p_{E}, \tilde \eta \in \partial \rho_\lambda(\theta)}
        }\nonumber\\
    &= \cbr{
        \sqrt{n} P_n \dot \ell_{\Mstar} + \gamma \hat W_n + r_n^A\rbr{\MestP} \in \cbr{ H_n \sqrt{n} \Mstar - H_n \theta - \tilde \eta : \theta \in \RR^p_{E}, \tilde \eta \in \partial \rho_\lambda(\theta)}
        },\nonumber
\end{align}
where the fifth equality uses \cref{cond:unique_solution} and \cref{eq:penalized_taylor}.

The set $\cbr{ H_n \theta + \tilde \eta : \theta \in \RR^p_{E}, \tilde \eta \in \partial \rho_\lambda(\theta)}$ is the union of finitely many disjoint measurable and convex sets by \cref{cond:penalty_union_convex}. Thus, the set $\cbr{ H_n \sqrt{n} \Mstar - H_n \theta - \tilde \eta : \theta \in \RR^p_{E}, \tilde \eta \in \partial \rho_\lambda(\theta)}$ is too, as merely an affine transformation. Thus, we can characterize the selection event as
\begin{equation*}
    \cbr{
        \supp\rbr{\MestP} = E
        }
    = \cbr{ \hat Z_n^{(1)} + 
        r_n^A\rbr{\MestP} \in A_n^E
        },
\end{equation*}
where $\hat Z_n^{(1)} = \sqrt{n} P_n \dot\ell_{\Mstar} + \gamma \hat W_n$ by definition in \cref{eq:thinning_core_stat}, and
\begin{equation}\label{eq:An_def}
    A_n^E := \cbr{ H_n \sqrt{n} \Mstar - H_n \theta - \tilde \eta : \theta \in \RR^p_{E}, \tilde \eta \in \partial \rho_\lambda(\theta)}
\end{equation}
is the union of finitely many disjoint convex sets.

\end{proof}

\subsection{Proof of \cref{lemma:inference_event}}\label{sec:proof_lemme_inference_event}

\begin{proof}
Recall that $\MestE$ is the model-$E$ minimizer defined in \cref{eq:MestE}. Let 
\begin{equation}\label{eq:Mest_app}
    \Mest := \rbr{\MestE, 0_{|E|}}
\end{equation}
be the expanded $p$-dimensional sparse vector.
By first-order optimality conditions, $\Mest$ solves
\begin{equation}\label{eq:inference_zero_E}
    \rbr{\sqrt{n} P_n \dot\ell_{\Mest} - \tfrac1\gamma \hat W_{n}}_E = 0_{\abr{E}}.
\end{equation}
Equivalently, since this places no restriction on the $E^c$ coordinates, we can write
\begin{equation}\label{eq:inference_zero}
    \sqrt{n} P_n \dot\ell_{\Mest} - \tfrac1\gamma \hat W_n \in 0_{\abr{E}} \times \RR^{\abr{E^c}}.
\end{equation}
Scaling \cref{cond:taylor_expansion} by $\sqrt{n}$ yields
\begin{equation}\label{eq:taylorE}
    \sqrt{n} P_n \dot\ell_{\Mest} = \sqrt{n} P_n \dot \ell_{\Mstar} +  \sqrt{n} H_n (\Mest - \Mstar) + r_n^B(\Mest),
\end{equation}
where $r_n^B(\Mest) = o_p\rbr{1}$.
Combining \cref{eq:taylorE} with the optimality condition in \cref{eq:inference_zero},
\begin{equation}\label{eq:inference_taylor}
    \sqrt{n} P_n \dot \ell_{\Mstar} - \tfrac1\gamma \hat W_n + \sqrt{n} H_n (\Mest - \Mstar) + r_n^B(\Mest) \in 0_{\abr{E}} \times \RR^{\abr{E^c}}.
\end{equation}

We now expand the coverage event
\begin{align}
    &\cbr{\xi^\top \theta^{E,*}_{n, j} \in \CI_n^{\xi, \alpha}\rbr{\MestE, \Shat}}\nonumber\\
    &= \cbr{\xi^\top \MstarE \in \sbr{\xi^\top \MestE \pm n^{-1/2}z_{1-\alpha/2}\rbr{\xi^\top \Shat\xi}^{1/2}}}\nonumber\\
    &= \cbr{\sqrt{n} \xi^\top  \rbr{\MstarE - \MestE} \in \sbr{\pm z_{1-\alpha/2}\rbr{\xi^\top \Shat\xi}^{1/2}}}\nonumber\\
    &= \cbr{\sqrt{n} \rbr{\MstarE - \MestE} \in \cbr{v \in \RR^{|E|}: \xi^\top v \in \sbr{\pm z_{1-\alpha/2}\rbr{\xi^\top \Shat\xi}^{1/2}}}}\nonumber\\
    &= \cbr{\sqrt{n} \rbr{\Mstar - \Mest} \in \cbr{v \in \RR^{p}: \xi^\top v_{E} \in \sbr{\pm z_{1-\alpha/2}\rbr{\xi^\top \Shat\xi}^{1/2}}}}\nonumber\\
    &= \cbr{\sqrt{n} H_n \rbr{\Mstar - \Mest} \in \cbr{v \in \RR^{p}: \xi^\top \rbr{H_{n, E, E}}^{-1} v_{E} \in \sbr{\pm z_{1-\alpha/2}\rbr{\xi^\top \Shat\xi}^{1/2}}}}\nonumber\\
    &= \cbr{\sqrt{n}  P_n \dot\ell_{\MstarE} - \tfrac1\gamma \hat W_n + r_n^B\rbr{\Mest} \in \cbr{v \in \RR^{p}: \xi^\top \rbr{H_{n, E, E}}^{-1} v_{E} \in \sbr{\pm z_{1-\alpha/2}\rbr{\xi^\top \Shat\xi}^{1/2}}}}\nonumber\\ 
    &=\cbr{
    \hat Z_n^{(2)} + r_n^B\rbr{\Mest}
    \in B_n^E},\label{eq:inference_equiv}
\end{align}
where the fourth equality makes use of the definitions of $\Mstar$ in \cref{eq:Mstar} and $\Mest$ in \cref{eq:Mest_app}, $\hat Z_n^{(2)} = \sqrt{n} P_n \dot\ell_{\Mstar} - \tfrac1\gamma \hat W_n$ by definition in \cref{eq:thinning_core_stat}, and
\begin{equation}\label{eq:Bn_def}
    B_n^E := \cbr{v \in \RR^{p}: \xi^\top \rbr{H_{n, E, E}}^{-1} v_{E} \in \sbr{\pm z_{1-\alpha/2}\rbr{\xi^\top \Shat\xi}^{1/2}}}.
\end{equation}
Note that $B_n^E$ is a convex set since it the pre-image of a convex set with respect to an affine function.
\end{proof}

\subsection{Proof of \cref{thm:valid_inf_grad}}
\begin{proof}

Recalling the definition of $B_n^E$ in \cref{eq:Bn_def}, we lower bound the conditional probability
\begin{align}
    &\Pr \rbr{ \xi^\top \MstarE \in \CI_n^{\xi, \alpha}\rbr{\MestE, \Shat} \middle|\, \supp\rbr{\MestP} = E}\nonumber\\
    &=
    \Pr \rbr{ \xi^\top \MstarE \in \CI_n^{\xi, \alpha}\rbr{\MestE, \Shat} \middle|\, \supp\rbr{\MestP} = E} - \Pr\rbr{Z_n^{(2)} \in B_n^E}\nonumber\\
    &\quad + \Pr\rbr{Z_n^{(2)} \in B_n^E}
    - \Pr\rbr{\hat Z_n^{(2)} + r_n^B\rbr{\Mest} \in B_n^E} + \Pr\rbr{\hat Z_n^{(2)} + r_n^B\rbr{\Mest} \in B_n^E}\nonumber\\
    &\geq \Pr\rbr{\hat Z_n^{(2)} + r_n^B\rbr{\Mest} \in B_n^E} - \abr{\Pr\rbr{Z_n^{(2)} \in B_n^E} - \Pr\rbr{\hat Z_n^{(2)} + r_n^B\rbr{\Mest} \in B_n^E}}\label{eq:to_bound2}\\
    &\quad-\abr{\Pr \rbr{ \xi^\top \MstarE \in \CI_n^{\xi, \alpha}\rbr{\MestE, \Shat} \middle|\, \supp\rbr{\MestP} = E}
    - \Pr\rbr{Z_n^{(2)} \in B_n^E}}\label{eq:to_bound1}.
\end{align}
By \cref{lemma:inference_event}, as in \cref{eq:inference_equiv}, and \cref{cond:sandwich_est},
\begin{equation*}
    \lim_{n \to \infty}\Pr\rbr{\hat Z_n^{(2)} + r_n^B\rbr{\Mest} \in B_n^E}
    = \lim_{n \to \infty} \Pr\rbr{\xi^\top \theta^{E,*}_{n, j} \in \CI_n^{\xi, \alpha}\rbr{\MestE, \Shat}} \geq 1-\alpha.
\end{equation*}
It remains to show that each of the remaining terms in \cref{eq:to_bound2} and \cref{eq:to_bound1} are $o(1)$.

\paragraph{Bounding \cref{eq:to_bound1}}
We apply \cref{lemma:conditional_berry_esseen} to the sets $A_n^E$ introduced in \cref{lemma:selection_event}, and $B_n^E$ introduced in \cref{lemma:inference_event}.
For clarity, we begin by restating \cref{lemma:conditional_berry_esseen} using the notation of \cref{lemma:selection_event} and \cref{lemma:inference_event}.

Since $A_n^E$ (defined in \cref{eq:An_def}) is the union of $2^{|E|}$ disjoint convex sets, and $B_n^E$ (defined in \cref{eq:Bn_def}) is a convex set, $A_n^E \times B_n^E$ can be written as $K = 2^{|E|}$ convex sets.
Let $\Delta > 0$, and $\varepsilon > 0$ be less than the smallest inradii of these $K$ sets, defined in \cref{eq:inradius}. 
For brevity, define the terms
\begin{align}
    E_1 &:= p^{1/4} n^{-1/2}\eigmin{\Sigma_n}^{-3/2} + \Delta + \Pr\rbr{\nbr{\hat \Sigma_n - \Sigma_n}_F > \Delta \eigmin{\Sigma_n}},\label{eq:E1}\\
    E_2 &:= O\rbr{E_1 + p^{1/2}\varepsilon\rbr{\eigmin{\Sigma_n}^{1/2} + \eigmin{\Sigma_n}^{-1/2}}} + \Pr\rbr{\nbr{r_n^A}_2 > \tfrac\varepsilon2} + \Pr\rbr{\nbr{r_n^B}_2 > \tfrac\varepsilon2},\label{eq:E2}\\
    E_3 &:= \Pr\rbr{Z^{(1)}_n \in A_n^{-\varepsilon}} - O\rbr{E_1} - \Pr\rbr{\nbr{r_n^A}_2 > \varepsilon},\label{eq:E3}
\end{align}
where $r_n^A \rbr{\MestP}$ defined in \cref{eq:taylor}, $r_n^B \rbr{\Mest}$ defined in \cref{eq:taylorE}, and $\rbr{A_n^E}^{-\varepsilon}$ is the $\varepsilon$-shrinkage (defined in \cref{eq:epsilon_shrinkage}) of $A_n^E$.
If Condition \ref{cond:berry_esseen} holds, then
\begin{align}\label{eq:master_stat_app}
    \abr{
    \Pr \rbr{\hat Z^{(2)}_n + r_n^B \rbr{\Mest} \in B_n^E \middle| \hat Z^{(1)}_n + r_n^A \rbr{\MestP} \in A_n^E}
    - \Pr \rbr{Z^{(2)}_n \in B_n^E}
    }
    \leq
    \frac{
    KE_2}
    {\max\cbr{E_3, KE_2}}.
\end{align}

It remains to show that the bound vanishes. First, we establish that for large enough $n$ the inradii are bounded above zero and hence we can fix an $\varepsilon > 0$. To show this, it suffices to show that $B_n^E$ converges to a fixed set, and that each of the $2^K$ constituent sets of $A_n^E$ converge to fixed sets with non-empty interiors.
 
To begin, recall from \cref{eq:Bn_def} that
\begin{equation*}
    B_n^E := \cbr{v \in \RR^{p}: \xi^\top \rbr{H_{n, E, E}}^{-1} v_{E} \in \sbr{\pm z_{1-\alpha/2}\rbr{\xi^\top \Shat\xi}^{1/2}}}.
\end{equation*}
By \cref{ass:existing_limits}, $H_{n, E, E} \to H_{E, E}$ and $S_n^E \to S^E := \rbr{\hat H_{E, E}}^{-1} \hat \Sigma_{E, E} \rbr{\hat H_{E, E}}^{-1}$, and thus
\begin{equation*}
    B_n^E \to B^E := \cbr{v \in \RR^{p}: \xi^\top \rbr{H_{E, E}}^{-1} v_{E} \in \sbr{\pm z_{1-\alpha/2}\rbr{\xi^\top S^E \xi}^{1/2}}}.
\end{equation*}
In the $E^c$ coordinates, the elements of $B^E$ are unrestricted. By \cref{ass:existing_limits}, $H$ and $\Sigma$ are positive definite and so $\sbr{\pm z_{1-\alpha/2}\rbr{(1 + \tfrac{1} {\gamma^2})\xi^\top S^E\xi}^{1/2}}$ contains a non-empty open set. Since the maximum eigenvalue of $\rbr{H_{E, E}}^{-1}$ is also bounded away from zero by \cref{ass:existing_limits}, the elements of $B^E$ contain a non-empty open set in the $E$ coordinates with inradius bounded away from zero. Thus, $B^E$ contains a non-empty open set.

Now, recall from \cref{eq:An_def} that
\begin{equation*}
    A_n^E := \cbr{ H_n \sqrt{n} \Mstar - H_n \theta - \tilde \eta : \theta \in \RR^p_{E}, \tilde \eta \in \partial \rho_\lambda(\theta)}.
\end{equation*}
By Assumption \ref{ass:existing_limits}(i) and \ref{ass:existing_limits}(iii),
\begin{equation*}
    A_n^E \to A^E := \cbr{ H \theta^{E, *} - H \theta - \tilde \eta : \theta \in \RR^p_{E}, \tilde \eta \in \partial \rho_\lambda(\theta)},
\end{equation*}
which by \cref{cond:penalty_union_convex} can be written as the union of finitely many disjoint convex sets with non-empty interiors.

It remains to show that the upper bound in \cref{eq:master_stat_app} asymptotically vanishes. We examine each of the three constituent terms $E_1$, $E_2$, and $E_3$. First, recall from \cref{eq:E1} that
\begin{equation*}
    E_1 := p^{1/4} n^{-1/2}\eigmin{\Sigma_n}^{-3/2} + \Delta + \Pr\rbr{\nbr{\hat \Sigma_n - \Sigma_n}_F > \Delta \eigmin{\Sigma_n}}.
\end{equation*}
Since $\eigmin{\Sigma_n}^{-3/2} \to \eigmin{\Sigma}^{-3/2} < \infty$ by \cref{ass:existing_limits}, clearly $p^{1/4} n^{-1/2}\eigmin{\Sigma_n}^{-3/2} \to 0$. Similarly, since $\eigmin{\Sigma_n} \to \eigmin{\Sigma} > 0$ by \cref{ass:existing_limits}, by \cref{cond:Sigma_est} and Markov's inequality, for any $\Delta > 0$
\begin{equation*}
    \Pr\rbr{\nbr{\hat \Sigma_n - \Sigma_n}_F > \Delta \eigmin{\Sigma_n}} = o(1).
\end{equation*}
Since $\Delta$ was arbitrary, we can choose a sequence $\Delta_n$ such that
\begin{equation*}
    \Delta_n + \Pr\rbr{\nbr{\hat \Sigma_n - \Sigma_n}_F > \Delta_n \eigmin{\Sigma_n}} = o(1).
\end{equation*}

Now recall from \cref{eq:E2} that
\begin{equation*}
    E_2 := O\rbr{E_1 + p^{1/2}\varepsilon\rbr{\eigmin{\Sigma_n}^{1/2} + \eigmin{\Sigma_n}^{-1/2}}} + \Pr\rbr{\nbr{r_n^A}_2 > \tfrac\varepsilon2} + \Pr\rbr{\nbr{r_n^B}_2 > \tfrac\varepsilon2}.
\end{equation*}
By \cref{lemma:selection_event} and \cref{lemma:inference_event}, for any $\varepsilon > 0$,
\begin{equation*}
    \Pr\rbr{\nbr{r_n^A \rbr{\MestP}}_2 > \tfrac\varepsilon2} + \Pr\rbr{\nbr{r_n^B \rbr{\Mest}}_2 > \tfrac\varepsilon2} = o(1).
\end{equation*}
Since $\varepsilon$ was arbitrarily small and $\eigmin{\Sigma_n} \to \eigmin{\Sigma}$ by \cref{ass:existing_limits}, $\varepsilon\rbr{\eigmin{\Sigma_n}^{1/2} + \eigmin{\Sigma_n}^{-1/2}} \to 0$ and so $E_2 \to 0$. 

Lastly, recall from \cref{eq:E3} that
\begin{equation*}
    E_3 := \Pr\rbr{Z^{(1)}_n \in A_n^{-\varepsilon}} - O\rbr{E_1} - \Pr\rbr{\nbr{r_n^A}_2 > \varepsilon}.
\end{equation*}
By the logic used earlier to analyze $E_2$, $O\rbr{E_1} - \Pr\rbr{\nbr{r_n^A}_2 > \varepsilon} \to 0$, and so it remains to study $\Pr\rbr{Z^{(1)}_n \in \rbr{A_n^E}^{-\varepsilon}}$. Since $\varepsilon$ is less than the inradius of the limit $A^E$ of $A_n^E$ and $\Sigma_n$ converges to the positive definite $\Sigma$. It follows that
\begin{equation*}
    \liminf_{n \to \infty} \Pr\rbr{Z^{(1)}_n \in \rbr{A_n^E}^{-\varepsilon}} > 0.
\end{equation*}
This establishes that the upper bound in \cref{eq:master_stat_app} asymptotically vanishes:
\begin{equation*}
    \lim_{n \to \infty} \frac{E_2}{\max\cbr{E_2, E_3}} = 0.
\end{equation*}

\paragraph{Bounding \cref{eq:to_bound2}}

We need to show that
\begin{equation*}
    \lim_{n \to \infty}\abr{\Pr\rbr{\hat Z_n^{(2)} + r_n^B\rbr{\Mest} \in B_n^E} - \Pr\rbr{Z_n^{(2)} \in B_n^E}} = 0.
\end{equation*}

\cref{eq:boundo_on_b7} upper bounded this absolute difference with $E_2$ defined in \cref{eq:E2}. As previously established, $\lim_{n \to \infty}E_2 = 0$.

\end{proof}

\subsection{Proof of \cref{thm:valid_inf_outcome}}

Under \cref{cond:affine}, it is clear to see that
\begin{equation*}
    \ell_\theta\rbr{Y_i - \tfrac1\gamma \bar W_{n, i}} = \ell_\theta(Y_i) - \tfrac1\gamma \theta^\top A_i \bar W_{n, i},
\end{equation*}
where $A_i \bar W_{n, i} \sim \Norm\rbr{0, A_i \bar\Sigma_{n, i} A_i^\top }$.
Thus, analyzing $\MestEoutcome$,
\begin{align*}
    \MestEoutcome
    &:= \argmin_{\theta \in \RR^{\abr{E}}} \cbr{\frac1n \sum_{i=1}^n \ell_{\rbr{\theta, 0_{p-\abr{E}}}}\rbr{Y_i - \tfrac1\gamma  \bar W_{n, i}}}\\
    &= \argmin_{\theta \in \RR^{\abr{E}}} \cbr{\frac1n \sum_{i=1}^n 
    \sbr{\ell_{\rbr{\theta, 0_{p-\abr{E}}}}(Y_i) - \tfrac1\gamma \rbr{\theta, 0_{p-\abr{E}}}^\top A_i \bar W_{n, i}}}\\
    &= \argmin_{\theta \in \RR^{\abr{E}}} \cbr{\frac1n \sum_{i=1}^n 
    \ell_{\rbr{\theta, 0_{p-\abr{E}}}}(Y_i) - \rbr{\theta, 0_{p-\abr{E}}}^\top \frac1n \sum_{i=1}^n \tfrac1\gamma A_i \bar W_{n, i}}\\
    &=\argmin_{\theta \in \RR^{\abr{E}}} \cbr{P_n \ell_{\rbr{\theta, 0_{p-\abr{E}}}} - \tfrac1\gamma n^{-1/2} \rbr{\theta, 0_{p-\abr{E}}}^\top \hat W_{n}}\\
    &=\argmin_{\theta \in \RR^{\abr{E}}} \cbr{P_n \ell_{\rbr{\theta, 0_{p-\abr{E}}}} - \tfrac1\gamma n^{-1/2} \theta^\top \hat W_{n, E}}
    =: \MestE,
\end{align*}
if $\hat W_n = n^{-1/2} \sum_{i=1}^n A_i \bar W_{n, i}$. A similar proof shows that $\MestP = \MestPoutcome$. By equality, we can apply \cref{thm:valid_inf_grad}.

\subsection{Proof of \cref{prop:taylor_expansion_exists}}
\begin{proof}
    We first establish that $\MestP$, defined in \cref{eq:MestP}, converges to the unique (by assumption) population quantity
    \begin{equation}\label{eq:MstarP}
        \MstarP := \argmin_{\theta \in \Theta} \EE P_n \ell_\theta,
    \end{equation}
    and that we have consistency of $\MestE \inprobability \MstarE$, where these two quantities were defined in \cref{eq:MestE} and \cref{eq:MstarE}.
    
    By Theorem 5.7 of \citet{vaart_asymptotic_1998}, since (i) our estimands are unique, (ii) we have uniform convergence, and $\MestP$ is a near minimizer of $\EE P_n \ell_\theta$ (as the penalty and randomized terms are $o_p(1)$), it follows that $\MestE \inprobability \MstarE$ and $\MestP \inprobability \MstarP$.

    Having established consistency, we are ready to show that \cref{cond:taylor_expansion} is met. The following results and conditions are standard (see \citet[Theorem 5.41]{vaart_asymptotic_1998}). By existence of third derivatives, for any $\theta$ we can construct the following Taylor expansion of $P_n \dot\ell_{\theta}$ around $\Mstar$ defined in \cref{eq:Mstar}:
    \begin{align}
        P_n \dot\ell_{\theta}
        &=P_n \dot \ell_{\Mstar} + \sbr{P_n \ddot\ell_{\Mstar}} (\theta - \Mstar) + \tfrac12 (\theta - \Mstar)^\top \sbr{P_n \dddot\ell_{\tilde \theta_n}} (\theta - \Mstar)\label{eq:taylor1}\\
        &=  P_n \dot \ell_{\Mstar} + H_n (\theta - \Mstar) + r_n(\theta),
        \label{eq:taylor_selection}
    \end{align}
    where
    \begin{equation}\label{eq:remainder}
        r_n(\theta) := \rbr{P_n \ddot\ell_{\Mstar} - H_n + \tfrac12 (\theta - \Mstar)^\top \sbr{P_n \dddot\ell_{\tilde \theta_n}}} (\theta - \Mstar)
    \end{equation}
    for some $\tilde \theta_n$ which lies on the line segment between $\theta$ and $\Mstar$.

    We now characterize terms in the above expressions.
    By the law of large numbers, $P_n \ddot\ell_{\Mstar} - H_n = o_p(1)$ where $H_n := \EE P_n \ddot\ell_{\Mstar}$. Lastly, provided that $\dddot \ell_{\tilde \theta_n}(Y_i)$ is dominated by some integrable function $\dddot\ell(Y_i)$ for any $\tilde \theta_n$ in some neighborhood $\Ncal_n$,
    \begin{equation}\label{eq:neighborhood_bound}
       \sup_{\tilde \theta_n \in \Ncal_n} \nbr{P_n \dddot\ell_{\tilde \theta_n}}_F
       \leq \nbr{P_n \dddot\ell}_F = O_p(1)
    \end{equation}
    by the law of large numbers. \cref{eq:neighborhood_bound} will clearly hold if such a $\dddot \ell_{\tilde \theta_n}(Y_i)$ exists for all $\tilde \theta_n \in \Theta$. Alternatively, \cref{eq:neighborhood_bound} will hold with probability tending to one if $\Ncal_n$ is an appropriate neighborhood around $\MstarP$ in the case of a Taylor expansion of $P_n\dot\ell_{\MestP}$, or neighborhood around $\Mstar$ in the case of $P_n\dot\ell_{\Mest}$.

    \paragraph{Analyzing the inference estimator $\Mest$:}

    The first order optimality condition of \cref{eq:MestE} implies that
    \begin{equation*}
        \sqrt{n} \rbr{P_n \dot\ell_{\Mest} - \tfrac1\gamma \hat W_{n}}_E = 0.
    \end{equation*}
    Replacing $P_n \dot\ell_{\Mest}$ with its Taylor expansion in \cref{eq:taylor_selection}, we arrive at the classical results $\sqrt{n}\rbr{\MestE - \MstarE} = O_p(1)$. This turn implies that $\Mest - \Mstar = O_p\rbr{n^{-1/2}}$, and furthermore that
    \begin{equation*}
        r_n\rbr{\Mest} := \rbr{P_n \ddot\ell_{\Mstar} - H_n + \tfrac12 (\Mest - \Mstar)^\top \sbr{P_n \dddot\ell_{\tilde \theta}}} (\Mest - \Mstar) = o_p\rbr{n^{-1/2}}.
    \end{equation*}

    \paragraph{Analyzing the selection estimator $\MestP$:}

    We now turn to analyzing $r_n\rbr{\MestP}$. By the first order optimality condition of \cref{eq:MestP},
    \begin{equation*}
        \sqrt{n} P_n \dot\ell_{\MestP} + \gamma \hat W_n + \eta = 0,
    \end{equation*}
    for some $\eta \in \partial \rho_\lambda\rbr{\MestP}$.
    
    As a first step, show that under a local asymptotic setting, $\sqrt{n} \MestP = O_p(1)$. We replace $P_n \dot\ell_{\MestP}$ with its Taylor expansion around $\MstarP$ in \cref{eq:MstarP}:
    \begin{equation*}
        \sqrt{n} P_n \dot\ell_{\MstarP} + \gamma \hat W_n + \eta +  \sqrt{n} (\MestP - \MstarP) \sbr{P_n \ddot\ell_{\MstarP} + \tfrac12 (\MestP - \MstarP)^\top \sbr{P_n \dddot\ell_{\tilde \theta_n'}}} = 0
    \end{equation*}
    for some $\tilde \theta_n'$ between $\MestP$ and $\MstarP$. Since $\hat W_{n} = O_p(1)$ and $\eta$ is bounded, under typical regularity conditions $\sqrt{n} P_n \dot\ell_{\MstarP} = O_p(1)$ and it follows that $\sqrt{n} (\MestP - \MstarP) = O_p(1)$. Our assumed local alternative setting where $\sqrt{n} \MstarP = O(1)$ implies that $\sqrt{n} \MestP = O_p(1)$.

    We now demonstrate that $r_n\rbr{\MestP} = o_p\rbr{n^{-1/2}}$.
    We replace $P_n \dot\ell_{\MestP}$ with its Taylor expansion in \cref{eq:taylor1}:
    \begin{equation*}
        \sqrt{n} P_n \dot\ell_{\Mstar} + \gamma \hat W_n + \eta +  \sqrt{n} (\MestP - \Mstar) \sbr{P_n \ddot\ell_{\Mstar} + \tfrac12 (\MestP - \Mstar)^\top \sbr{P_n \dddot\ell_{\tilde \theta_n}}} = 0
    \end{equation*}
    for some $\tilde \theta_n$ between $\MestP$ and $\Mstar$. By \cref{ass:existing_limits}[i], we know that $\sqrt{n} \Mstar = O(1)$ and thus that $\sqrt{n} (\MestP - \Mstar) = O_p(1)$.
    This in turn implies that
    \begin{equation*}
        \sqrt{n} r_n\rbr{\MestP} := \sqrt{n} (\MestP - \Mstar) \rbr{P_n \ddot\ell_{\Mstar} - H_n + \tfrac12 (\MestP - \Mstar)^\top \sbr{P_n \dddot\ell_{\tilde \theta_n}}} = o_p\rbr{1}
    \end{equation*}
    as desired, since $P_n \ddot\ell_{\Mstar} - H_n = o_p(1)$ by the law of large numbers and $P_n \dddot\ell_{\tilde \theta_n} = O_p(1)$.

\end{proof}

\subsection{Proof of \cref{prop:is_valid_penalty}}

\begin{proof}
    Recall that for $\rho_\lambda(\theta) := \lambda\nbr{\theta}_1$,  for any $j \in [p]$,
    \begin{equation*}
        [\partial \rho_\lambda(\theta)]_j :=
        \begin{cases}
            \lambda\,\sign(\theta_j) & \text{if } \theta_j \neq 0\\
          [-\lambda, \lambda] & \text{if } \theta_j = 0
        \end{cases}.
    \end{equation*}

    For a fixed $\lambda$, this is clearly bounded and thus satisfies \cref{cond:penalty_union_convex}(i). It remains to show that \cref{cond:penalty_union_convex}(ii) holds, i.e.,
    for any positive definite matrix $H \in \RR^{p \times p}$,
    \begin{equation}\label{eq:L_1_set}
        \cbr{H\theta + \eta : \supp\rbr{\theta} = E, \eta \in \partial \rho_\lambda(\theta)}
    \end{equation}
    has a non-empty interior and can be written as the union of finitely many disjoint convex sets.

    To begin, let $s_E \in \cbr{-1, 1}^{|E|}$ be one of the $2^{|E|}$ possible combinations of signs of $\theta_E$, $\RR_{E, s_E}^p := \cbr{\theta : \supp(\theta) = E, \sign(\theta_E) = s_E}$ be an orthant, and $\Ecal_{E, s_E} := \cbr{\eta : \eta \in \partial \rho_\lambda(\theta), \theta \in \RR_{E, s_E}^p}$ be a collection of subgradients.
    We can then write
    \begin{align*}
        &\cbr{H\theta + \eta : \supp\rbr{\theta} = E, \eta \in \partial \rho_\lambda(\theta)}\\
        &= \bigcup_{s_E \in \cbr{-1, 1}^{|E|}} \cbr{H\theta + \eta : \theta \in \RR_{E, s_E}^p, \eta \in \partial \rho_\lambda(\theta)} \\
        &= \bigcup_{s_E \in \cbr{-1, 1}^{|E|}} \cbr{H\theta : \theta \in \RR_{E, s_E}^p} + \Ecal_{E, s_E},
    \end{align*}
    where the final sum denotes the Minkowski sum, i.e., the pairwise sum of all points in the two sets.

    First, we can see that $\cbr{H\theta : \theta \in \RR_{E, s_E}^p}$ is an affine transformation of a convex set, $\Ecal_{E, s_E}$ is a convex set, and so their Minkowski sum is a convex set. The union encompasses $2^{|E|}$ such convex sets. To show that these sets are disjoint, for sake of contradiction assume that there is a point lying in two such sets. This implies that there exists some $s_E \neq s_E' \in \cbr{-1, 1}^{|E|}$ and $\theta \in \RR^p_{E, s_E}, \theta' \in \RR^p_{E, s_E'}$ such that $H_{E, E} \theta_E + \lambda s_E = H_{E, E} \theta_E' + \lambda s_E'$. Rearranging,
    \begin{equation*}
        H_{E, E} (\theta_E - \theta_E') = \lambda(s_E' - s_E).
    \end{equation*}
    Since $H$ is positive definite, $(\theta_E - \theta_E')^\top H_{E, E} (\theta_E - \theta_E') > 0$. However, $\lambda(\theta_E - \theta_E')^\top (s_E' - s_E) < 0$ by definition of the signs. Thus, we have a contradiction and so \cref{eq:L_1_set} is the union of $2^{|E|}$ disjoint convex sets.

    It remains to show that these sets have non-empty interiors, for which it suffices to show that $\cbr{H\theta : \theta \in \RR_{E, s_E}^p} + \Ecal_{E, s_E}$ contains a non-empty open set.
    In the $E$ coordinates of $\cbr{H\theta : \theta \in \RR_{E, s_E}^p}$, since $\RR_{E, s_E}^p$ contains a non-empty open set, and $H$ is positive definite, $\cbr{(H\theta)_E : \supp\rbr{\theta} = E, \sign(\theta_E) = s_E}$ also contains a non-empty open set. In the $E^c$ coordinates, we can clearly see that $\Ecal_{E, s_E}$ contains a non-empty open set and thus the Minkowski sum $\cbr{(H\theta) : \theta \in \RR_{E, s_E}^p} + \Ecal_{E, s_E}$ does too.
\end{proof}

\section{Supplemental results}

\subsection{An algorithm implementing \cref{thm:valid_inf_grad} for $L_1$-penalized GLMs}\label{sec:app_grad_algo}

Here, we additionally provide the procedure for obtaining valid inference using Algorithm \cref{thm:valid_inf_grad}.

\begin{mybox}[label=algo:valid_inf_grad]{implementing \cref{thm:valid_inf_grad} for $L_1$-penalized GLMs}
\textbf{Input:} Data $\cbr{Y_i, X_i}$, penalty $\lambda$, information split $\gamma$.
\begin{enumerate}
    \item $\tilde\theta_n \leftarrow
    \argmin_{\theta \in \RR^{p}} P_n \ell_{\theta}$.
    \item Sample noise $\hat W_{n} \leftarrow \Norm\rbr{0, \frac{1}{n}\sum_{i=1}^n \hat \alpha\rbr{\tilde\theta_n} \,\ddot b \rbr{X_i^\top \tilde\theta_n} X_i  X_i^\top}$.
    \item Estimate $\MestP$ as in \cref{eq:MestP} and define $E := \supp\rbr{\MestP}$.
    \item Estimate $\MestE$ as in \cref{eq:MestE}.
    \item Compute the modified sandwich variance estimator:
    \begin{equation*}
        \Shat \leftarrow \left( H_{n, E, E}\rbr{\MestE} \right)^{-1} \tilde \Sigma_{n, E, E} \left( H_{n, E, E}\rbr{\MestE} \right)^{-1}
    \end{equation*}
    where
    \begin{equation}\label{eq:modified_sandwich}
        \tilde \Sigma_{n, E, E} \leftarrow \frac{1}{\gamma^2} \Var\rbr{\hat W_{n, E, E}} + \frac1n \sum_{i=1}^n \rbr{Y_i - \dot b(X_{i, E}^\top \MestE)}^2  X_{i, E} X_{i, E}^\top.
    \end{equation}
    \end{enumerate}
\textbf{Output:} A confidence interval centered at $\MestE$,
    \begin{equation*}
        \CI^{\xi, \alpha}_n\rbr{\MestE, \Shat} = \sbr{\xi^\top \MestE \pm n^{-1/2} z_{1-\alpha/2}\rbr{ \xi^\top \Shat \xi}^{1/2}}.
    \end{equation*}
\end{mybox}

Algorithm \ref{algo:valid_inf_grad} is quite similar to Algorithm \ref{algo:valid_inf_outcome}. A minor but key differences lies in the sandwich variance estimator $\Shat$ in Step 5. Algorithm \ref{algo:valid_inf_outcome} computes the ``meat" of the sandwich using squared residuals relative to the noisy outcomes $\cbr{Y_i - \tfrac{1}{\gamma} \tilde W_{n, i}}_{i=1}^n$; this is the classical approach in the M-estimation literature \citep{white_heteroskedasticity-consistent_1980}. However, \cref{thm:valid_inf_grad} does not use noisy outcomes and so there is not an immediately natural estimator $\Shat$ available satisfying \cref{cond:sandwich_est}. As a solution, we use the estimator in \cref{eq:modified_sandwich} which involves squared residuals relative to the observed outcomes $\cbr{Y_i}_{i=1}^n$, and then additionally factor in the variance of the noise $\hat W_n$. The justification is as follows.

In the setting of \cref{thm:valid_inf_outcome}, the canonical sandwich estimator in Algorithm \ref{algo:valid_inf_outcome} takes the form
\begin{equation*}
    \frac1n \sum_{i=1}^n \rbr{Y_i - \frac1\gamma \tilde W_{n, i} - \dot b(X_{i, E}^\top \MestE)}^2  X_{i, E} X_{i, E}^\top.
\end{equation*}
Expanding the quadratic, this equals
\begin{align*}
    \frac1n \sum_{i=1}^n \sbr{\rbr{Y_i - \dot b(X_{i, E}^\top \MestE)}^2 - 2 \frac1\gamma \tilde W_{n, i} Y_i + 2 \frac1\gamma \tilde W_{n, i} \dot b(X_{i, E}^\top \MestE) + \frac{1}{\gamma^2} \tilde W_{n, i}^2 }X_{i, E} X_{i, E}^\top.
\end{align*}
Examining the terms in the quadratic expansion, note that
\begin{equation*}
    \EE\sbr{\tilde W_{n, i} Y_i} = \EE\sbr{\tilde W_{n, i} \dot b(X_{i, E}^\top \MestE)} = 0,
\end{equation*}
since $\tilde W_{n, i}$ is scaled mean zero noise sampled independently from $Y_i$. Furthermore, 
\begin{equation*}
    \frac1n \sum_{i=1}^n X_{i, E} \EE\sbr{\tilde W_{n, i}^2} X_{i, E}^\top = \frac1n \sum_{i=1}^n X_{i, E} \Var\rbr{\tilde W_{n, i}} X_{i, E}^\top = \Var\rbr{\hat W_{n, E, E}},
\end{equation*}
since $\hat W_{n} := n^{-1/2} \sum_{i=1}^n X_i \tilde W_{n, i}$ as specified in \cref{thm:valid_inf_outcome}.
Thus,
\begin{align*}
    &\frac1n \sum_{i=1}^n \rbr{Y_i - \frac1\gamma \tilde W_{n, i} - \dot b(X_{i, E}^\top \MestE)}^2  X_{i, E} X_{i, E}^\top\\
    &=\frac{1}{\gamma^2} \Var\rbr{\hat W_{n, E, E}} + \frac1n \sum_{i=1}^n \rbr{Y_i - \dot b(X_{i, E}^\top \MestE)}^2  X_{i, E} X_{i, E}^\top + o_p(1),
\end{align*}
and so the the sandwich estimators in Algorithms \ref{algo:valid_inf_outcome} and \ref{algo:valid_inf_grad} are equivalent up to an $o_p(1)$ term.

\subsection{Additional simulation results}

While data thinning is not applicable to Bernoulli data, it can be applied to binomial data~\citep{neufeld_data_2024, dharamshi_generalized_2025}. In the case of logistic regression with discrete covariates, it may be possible to group the Bernoulli observations to form binomial observations. Binomial data thinning of these binomial observations thus permits valid inference after the lasso in a logistic regression model.

In \cref{fig:logistic_discrete}, we revisit the logistic regression experiment in \cref{fig:logistic}. However, for each sample size $n$, we duplicate the same unique $50$ feature vectors to allow the observations to be binned. For example, when $n=200$, there are four observations for each of the $50$ unique feature vectors. We see that binomial data thinning and sample splitting perform similarly, while score thinning yields slightly narrower confidence intervals and lower FDR. Meanwhile, the fully conditional approach of \citet{huang_selective_2024} performs poorly at smaller sample sizes.

\begin{figure}[!htb]
    \centering
    \includegraphics[width=1.0\linewidth]{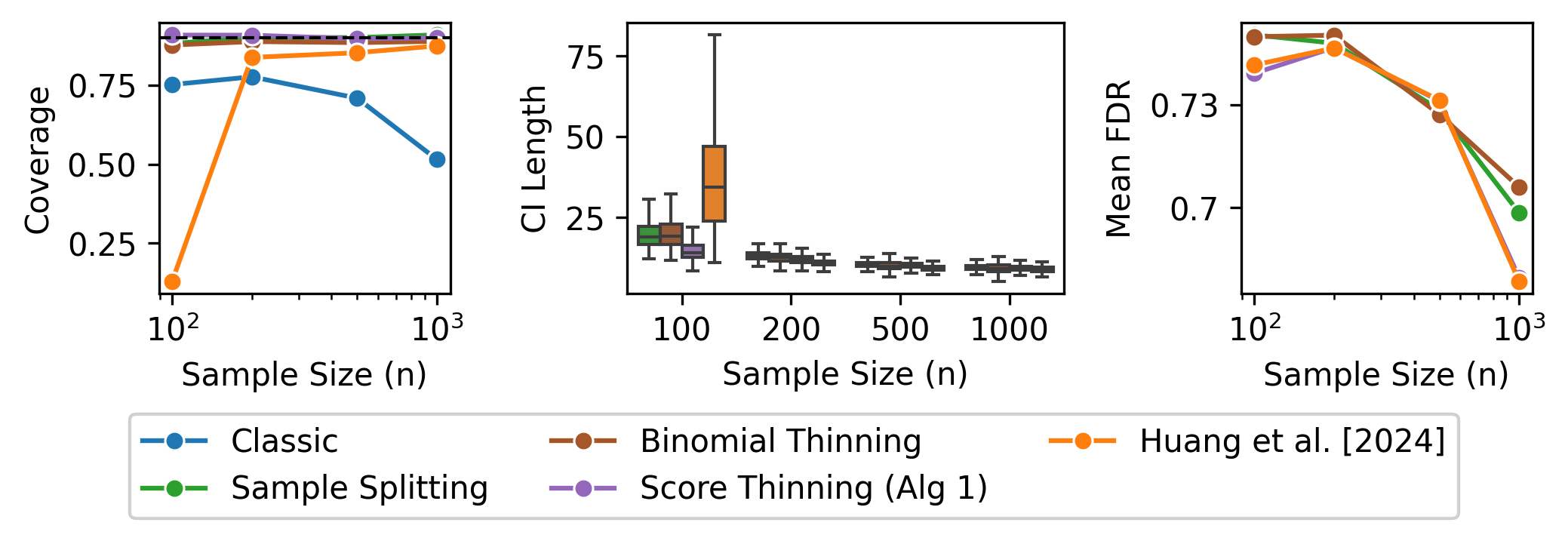}
    \caption{We simulate outcomes from a logistic regression model. There are $50$ unique feature vectors, and so binomial thinning is applicable to groups of size $n / 50$. Inference is conducted on the model selected using $L_1$-penalized logistic regression. Confidence intervals are computed for all coefficients in the selected model. Results are aggregated across $1000$ simulated datasets. Since the classical approach (which reuses the data for selection and inference) fails to attain valid coverage, we omit it from the center and right-hand panels.}
    \label{fig:logistic_discrete}
\end{figure}

\end{document}